\documentclass[10pt,a4paper]{article}
\usepackage{authblk}
\usepackage{geometry}
\usepackage{amsmath,amssymb,amsthm}
\usepackage{graphicx}
\usepackage{hyperref}
\usepackage{algorithm,algorithmic}

\usepackage{color}

\newtheorem{lemma}{Lemma}
\newtheorem{theorem}{Theorem}
\newtheorem{corollary}{Corollary}

\theoremstyle{definition}
\newtheorem{definition}{Definition}

\usepackage{refcount}

\newcommand{\ASY}{{\sc Asynch}}
\newcommand{\FSY}{{\sc Fsynch}}
\newcommand{\SSY}{{\sc Ssynch}}
\newcommand{\RSY}{{\sc Rsynch}}

\newcommand{\RR}{{\sc RoundRobin}}

\newcommand{\LU}{{\mathcal{LUMI}}} 
\newcommand{\LUM}{{\mathcal{LUMI}}} 
\newcommand{\FS}{{\mathcal{FSTA}}} 
\newcommand{\FC}{{\mathcal{FCOM}}} 
\newcommand{\OB}{{\mathcal{OBLOT}}} 
\newcommand{\N}{{\rm I\kern-.22em N}} 
\newcommand{\Z}{{\sf Z\kern-.42em Z}} 
\newcommand{\R}{R}

\newcommand{\CoP}{{\tt pred}}
\newcommand{\CoS}{{\tt suc}}
\newcommand{\LK}{{\mathit{Look}}}
\newcommand{\CP}{{\mathit{Comp}}}
\newcommand{\M}{{\mathit{Move}}}

\newcommand{\LC}{{\mathit{LC}}}

\newcommand{\CM}{{\mathit{CM}}}
\newcommand{\LCM}{{\mathit{LCM}}}

\newcommand{\T}{T}

\newcommand{\CLCv}{{\tt CLCv}}

\newcommand{\MLCv}{{\tt MLCv}}

\newcounter{Codeline}
\newcommand{\Newcodeline}{\setcounter{Codeline}{1}}
\newcommand{\Cl}{{\theCodeline}: \addtocounter{Codeline}{1}}
\newcommand{\crm}{\\}

\author[1]{Paola Flocchini}
\author[2]{Nicola Santoro}
\author[3]{Yuichi Sudo}
\author[3]{Koichi Wada\thanks{Corresponding Author: wada@hosei.ac.jp}}

\affil[1]{University of Ottawa, Canada}
\affil[2]{Carleton University, Canada}
\affil[3]{Hosei University, Japan}

\title{On Asynchrony, Memory, and Communication: Separations and Landscapes \thanks{{This research was partly supported by NSERC through the Discovery Grant program, by JSPS KAKENHI No. 20H04140, 20KK0232, 20K11685, 21K11748, and by JST FOREST Program JPMJFR226U.}}} 
\date{}

\begin{document}
\maketitle
\begin{abstract}
Research on distributed computing by a team of identical mobile computational entities, called robots, operating in a Euclidean space in $\mathit{Look}$-$\mathit{Compute}$-$\mathit{Move}$ 
($\mathit{LCM}$) cycles, has recently focused on better understanding how the computational power of robots depends on the interplay between their internal capabilities (i.e., persistent memory, communication), captured by the four standard
computational models ($\OB$, $\LU$, $\FS$, and $\FC$)
and the conditions imposed by the external environment, controlling the activation of the robots and their synchronization of their activities, perceived and modeled as an
adversarial scheduler.

We consider a set of adversarial asynchronous schedulers ranging from the classical {\em semi-synchronous}    (\SSY) and   {\em fully asynchronous}  (\ASY) settings, including  schedulers 
(emerging when studying the atomicity of the combination of operations in the $\mathit{LCM}$ cycles) whose adversarial power is in between those two.
We ask the question: what is the computational relationship between a model $M_1$ under adversarial scheduler $K_1$ ($M_1(K_1)$) and a model $M_2$ under scheduler $K_2$ ($M_2(K_2)$) ? For example, are the robots in $M_1(K_1)$ more powerful (i.e., they can solve more problems) than those in $M_2(K_2)$? 

We answer all these questions by providing, through 
cross-model analysis,
 a complete characterization of the computational relationship between the power of
the four models of robots under the considered asynchronous schedulers.
In this process, we also provide qualified answers to several open questions, including the outstanding one on the proper dominance of SSYNCH over ASYNCH in the case of unrestricted visibility.
\end{abstract}

\section{Introduction}\label{sec:Intro}

\subsection{Background}


\noindent {\bf Robot Models. }\ \
Since the seminal work of Suzuki and Yamashita~\cite{SY}, the studies of the computational issues arising in distributed systems of mobile computational entities, called {\em robots}, operating in a Euclidean space have focused on  
 identifying the minimal assumptions 
on {\em internal} capabilities of the robots (e.g., persistent memory, communication) and {\em external}
 conditions of the system  (e.g., synchrony, activation scheduler) 
that allow the entities to perform basic tasks and collectively solve given problems.

Endowed with computational, visibility and motorial capabilities,
the robots are  anonymous (\emph{i.e.}, indistinguishable from each other), uniform (\emph{i.e.}, run the same algorithm), and disoriented (\emph{i.e.}, they might not agree on a common coordinate system). Modeled as mathematical points in the 2D Euclidean plane in which they can freely move,
they operate in
$\mathit{Look}$-$\mathit{Compute}$-$\mathit{Move}$ ($\mathit{LCM}$) cycles. In each cycle, a robot ``$\mathit{Look}$s '' at its surroundings obtaining (in its current local coordinate system) a snapshot indicating the locations of the other robots. Based on this information, the robot executes its algorithm to ``$\mathit{Compute}$'' a destination, and then ``$\mathit{Move}$s'' towards the computed location.

In the (weakest and de facto) standard model, $\OB$, the robots are also
oblivious (\emph{i.e.}, they have no persistent memory of the past) and  silent (\emph{i.e.}, they have no explicit means of communication).  
Extensive investigations have been carried out to understand the computational 
 limitations and powers of  $\OB$  robots for basic coordination tasks such as Gathering (e.g., \cite{AP,AOSY,BDT,CDN,CFPS,CP,FPSW05,ISKIDWY,SY}), Pattern Formation (e.g., \cite{FPSW08,FYOKY,SY,YS,YUKY}), Flocking (e.g., \cite{CG,GP,SIW}); see 
 also the monograph \cite{FPS19} for a general account.

 The absence of persistent memory and the lack of explicit communication  critically restrict the computational capabilities 
of the $\OB$ robots, and limit the solvability of problems. 
These limitations are removed, to some extent, in the $\LU$ model of {\em luminous}
robots.
In this model, each robot is equipped with a constant-bounded amount of persistent\footnote{i.e.,  it is not automatically reset at the end of a cycle.}  memory, called {\em light},
whose value, called {\em color}, 
is visible to all robots. In other words, luminous robots can both
remember and communicate, albeit at a very limited level.
Since its introduction in  \cite{DFPSY}, the model has been the subject of several investigations focusing on the
 design of algorithms and the feasibility of problems for $\LU$ robots (e.g. \cite{sssbhagat17,DFPSY,LFCPSV,FSVY,HDT,OWD,OWK,SABM18,TWK,V}; see Chapter 11 of \cite{FPS19} for a recent survey).
 An important result is that, even if so limited, the simultaneous presence  of
 both persistent memory and communication renders
luminous robots strictly more powerful than oblivious robots \cite{DFPSY}.
This has in turns opened the question on the individual computational power of the two internal capabilities, memory and communication,
and motivated the investigations on 
two sub-models of $\LU$: the \textit{finite-state} robots denoted as $\FS$, where 
the robots have a constant-size persistent memory but are silent, 
and the \textit{finite-communication} robots denoted as $\FC$,  where robots can communicate a constant number 
of bits but are oblivious (e.g., see \cite{apdcm,BFKPSW22,FSVY,FSW19,OWD,OWK}).  

 

\noindent {\bf A/Synchrony.}\ \
All these studies in all those models have brought to light the crucial role played by two interrelated {\em external} factors:
 the level of synchronization and the activation schedule provided by the system.
 Like in other types of distributed computing systems, there are two different settings, the 
synchronous and the asynchronous ones.

In the {\em synchronous} (also called  {\em semi-synchronous})   (\SSY) setting, introduced in \cite{SY}, time   is divided into discrete 
intervals, called {\em rounds}. In each round, an arbitrary but nonempty subset 
of the robots is activated, and they simultaneously perform
exactly one  $\LK$-$\CP$-$\M$ cycle. The selection of which robots are 
activated at a given round is made by an adversarial scheduler, constrained only 
to be fair, i.e., every robot is activated infinitely often. 
Weaker synchronous adversaries have also been introduced and investigated. The most important and 
extensively studied is the {\em fully-synchronous} (\FSY) scheduler, which activates all the robots  in every round.
Other  interesting synchronous schedulers are
\RSY, 
where  the sets of robots activated in any two 
consecutive rounds are restricted to be disjoint, and it studied for its use to model energy-restricted robots \cite{BFKPSW22},  as well as
 the family of {\em sequential} schedulers  (e.g., 
 \RR),
where in each round only one robot is activated.

In the {\em asynchronous} setting (\ASY), introduced in  \cite{FPSW99}, 
 there is no common notion of time,  each robot is activated 
 independently of the others; it allows for finite but arbitrary delays between the $\LK$, $\CP$ and $\M$ phases, and each movement may take a finite but arbitrary amount of time. The duration of each cycle of a robot, as well as the decision of when a robot is activated,  are
controlled by an adversarial scheduler, constrained only to be fair,
i.e., every robot must be activated infinitely often.

Weaker adversaries are easily identified 
considering the atomicity of the combination of the $\LK$, $\CP$ and $\M$ stages.
In particular,  
if in every cycle the three operations are executed as a single atomic  instantaneous operation, this scheduler we shall call $LCM$-{\bf atomic}-\ASY\ coincides with \SSY. 
On the other hand, by combining fewer operations, two
asynchronous schedulers are identified  \cite{OWD}: $LC$-{\bf atomic}-\ASY, where the $\LK$ and $\CP$ operations are  a single atomic operation;  and  $CM$-{\bf atomic}-\ASY, where the $\CP$ and $\M$ operations are a single atomic operation.

Of independent interest is the restricted asynchronous adversary unable to
schedule the $\LK$  operation of a robot during the $\M$ operation of another.
The particular theoretical relevance of this scheduler, called 
$M$-{\bf atomic}-\ASY\ \cite{OWD} derives from the fact that one of the strongest debilitating effects of unrestricted asynchrony is precisely the fact that a robot, when looking,  cannot detect if another robot is still or moving.\\


\noindent {\bf Separators.}\ \
Like in other types of distributed systems, understanding the computational difference between (levels of) synchrony and asynchrony has been a primary research focus, first  in the $\OB$ model, and subsequently in the others. 

Indeed, one of the first results in the field has been the proof that in $\OB$ 
the simple problem  of two robots meeting at the same location, called {\em Rendezvous}(RDV), is unsolvable under \SSY\ \cite{SY} while easily solvable under \FSY, implying that fully synchronous $\OB$ robots  are strictly more powerful  than semi-synchronous ones.

Any problem that, like {\em Rendezvous}, proves the separation between the computational power of robots in two different settings is said to be a {\em separator}.
The quest has immediately been to determine if there are other problems in $\OB$ 
separating  \SSY\ from \FSY\ (i.e., the extent of their computational difference);
no other has been found so far.
Clearly more important and pressing has been the  question of whether there is any computational difference between synchrony and asynchrony. 
The quest for a problem separating \ASY\ from \SSY\ has been ongoing for more than two decades. Recently a separator has been found in the special case when the  visibility  range of the robots is limited \cite{KKNPS21}, leaving the existence of a separator  open
for the unrestricted case.

The quest for a separator in $\OB$   has been made more pressing since the result that no separation exists between \ASY\ and \SSY\ in the $\LU$ model \cite{DFPSY}; that is, the presence of a limited form of communication and memory is sufficient to completely overcome
the limitations imposed by asynchrony.
This result has motivated the investigation of the two submodels of $\LU$
where the robots are endowed with only  the limited form of persistent memory, $\FS$, 
or of communication, $\FC$.  While separation between fully synchrony and semi-synchrony has been shown to exist for both submodels \cite{apdcm,FSW19}, the more important question of whether  one of them is capable of overcoming asynchrony has not yet been answered; 
indeed, no separator between \SSY\ and \ASY\  has been found so far for either submodel.


\noindent {\bf Landscapes.} 
To understand the  impact that the factors of persistent memory and communication have  on the feasibility of problems, 
the main investigation tool has been  the
comparative analysis of the (new and/or existing) results obtained
for the same problems under the different four models $\OB, \FS, \FC, \LU$.
The same methodological tool can obviously be used also to 
establish the computational relationships between those models 
within a spectrum of schedulers, so to identify the relative powers of
those schedulers within each model. 

Through  this type of cross-model analysis, researchers have recently produced a comprehensive characterization of the computational relationship 
between the four models  with respect to the range of synchronous schedulers
 $<$\FSY, \RSY, \SSY$>$.
 creating a  comprehensive map  of the {\em synchronous landscape} 
 for  distributed systems of autonomous mobile robots
 in the four models \cite{apdcm,FSW19}. 

With respect to the (more powerful) {\em asynchronous} adversarial schedulers, 
ranging from $LCM$-{\bf atomic}-\ASY\ (i.e., \SSY) to \ASY,
very little is known to date on the computational 
power of persistent memory and of explicit communication in general,
and on the computational relationship between the four models in particular.
As mentioned, it is known that in $\LU$, robots have in \ASY\ the same computational  power as in \SSY\, and that  asynchronous luminous robots are strictly more powerful than oblivious synchronous robots \cite{DFPSY}.

Summarizing, while a comprehensive computational map has existed for the synchronous landscape, only disconnected fragments exist so far of the {\em asynchronous landscape}.

\subsection{Contributions}

In this paper, we analyze the computational relationship among
the four models $\OB$, $\FS$, $\FC$  and $\LU$,
under the range of asynchronous schedulers
$< LCM$-{\bf atomic}-\ASY, $LC$-{\bf atomic}-\ASY, $CM$-{\bf atomic}-\ASY, $M$-{\bf atomic}-\ASY, and \ASY $>$, 
 establishing a large variety  of results.
 Through these results,  we close several open problems, and 
 create a complete map of the asynchronous landscape for  distributed systems of autonomous mobile robots
 in the four models.

Among our contributions, we prove the existence of a separator 
 between \SSY\ and \ASY\  in the standard $\OB$ model for the unrestricted visibility case
by identifying  a simple natural problem,
{\tt Monotone Line Convergence} $($\MLCv$)$,
that separates \SSY\ from \ASY\ for $\OB$ robots.
This problem
requires two robots to convergence towards each other  monotonically (i.e., without ever increasing their distance) on the
line connecting them. We prove that this problem, trivially solvable 
in semi-synchronous systems, is however unsolvable if the system is asynchronous.


Because of this separation in $\OB$ on one hand,
and of the known absence of separation in  $\LU$ on the other, the next immediate question is whether either of $\LU$'s 
specific features (i.e., constant-sized communication and persistent memory) is strong enough alone to overcome asynchrony. In other words,
are there separators between \SSY\ and \ASY\  in $\FS$ ? in $\FC$ ?
In these regards, we provide a positive answer to both questions, thus proving
 that both features are needed to overcome asynchrony. 

The characterization of the computational relationship
between the four models with respect to the range of asynchronous schedulers is {\em complete}: for any two models, $M_1,M_2\in\{\OB,\FS, \FC, \LU \}$ and 
adversarial schedulers  $K_1,K_2\in\{ LCM$-{\bf atomic}-\ASY, $LC$-{\bf atomic}-\ASY, $CM$-{\bf atomic}-\ASY, $M$-{\bf atomic}-\ASY, \ASY $\}$ 
it is determined whether the computational power of (the robots in) $M_1$ under $K_1$  is stronger than, weaker than, equivalent to or orthogonal to (i.e., incomparable with) that of (the robots in) $M_2$ under $K_2$.

For example, we prove that 
for  $\FS$ (i.e., in presence of only limited internal persistent memory),  \SSY\ is computationally more powerful 
than {\sc Move}-{\bf atomic}-\ASY, which in turn is computationally 
more powerful than  \ASY. The several orthogonality (i.e., incomparability)  results include for example the fact that the combination
of asynchrony and limited persistent memory  is neither more nor less powerful than the combination
of synchrony and obliviousness.
Observe that to prove that a model under a specific scheduler is stronger than or orthogonal to
another model and scheduler (or same model and a different scheduler, or other model and same scheduler) requires to determine a problem solvable in one setting but not in the other.

Among the  equivalence of two models each under a specific scheduler,  we have proved that for  $\FC$ (i.e., in presence of only limited communication): the atomic combination of {\em Compute} and {\em Move} does not provide any gain with respect to complete asynchrony; on the other hand,
the atomic combination of {\em Look} and {\em Compute} completely overcomes asynchrony.
The proof of the equivalence 
has involved  designing a {\em simulation protocol} that allows to correctly execute any protocol for the first  model and scheduler into the other model and scheduler.

The resulting asynchronous landscape is shown in Figure \ref{fig:Relation-Diagram-SvsAbelow} 
where 
${S}$, ${A}$, ${A_{LC}}$ ${A_M}$, and ${A_{CM}}$  
 denote 
\SSY, \ASY, $\LC$-{\bf atomic}-\ASY, $M$-{\bf atomic}-\ASY, and $\CM$-{\bf atomic}-\ASY, respectively;
a box located higher than another indicates dominance unless they are connected by a dashed line,
which denotes orthogonality;  equivalence is indicated directly in the boxes.

\begin{figure}[tbh]
   \centering
\centering\includegraphics[keepaspectratio, width=0.7\textwidth]{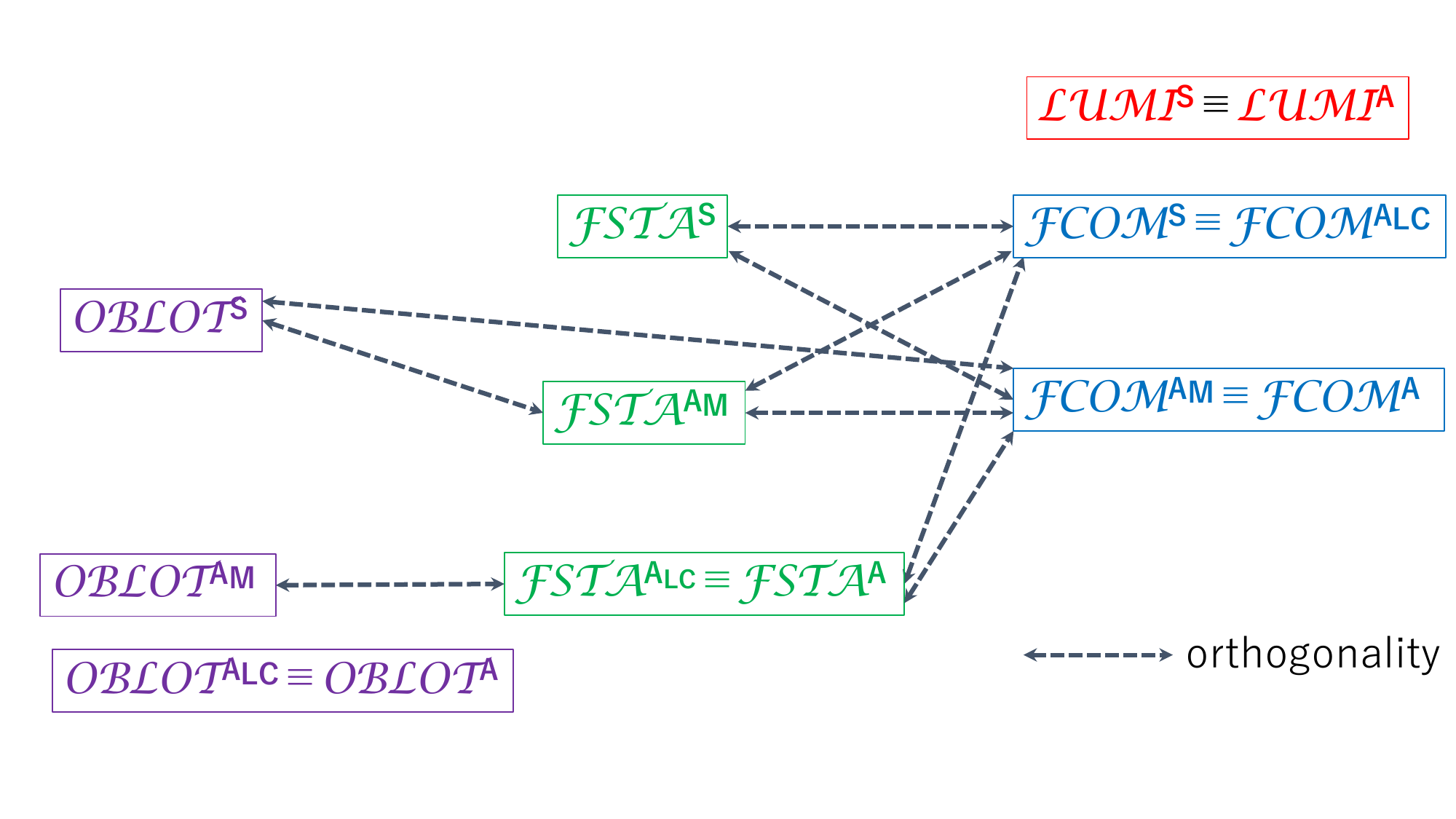}
   \caption{Asynchronous landscape of $\LU$, $\FC$,  $\FS$ and $\OB$.}
    \label{fig:Relation-Diagram-SvsAbelow}
\end{figure} 



\section{Models and Preliminaries}\label{sec:model}

\subsection{Robots}
We shall consider  a set  $R = \{ r_0 ,\cdots, r_{n-1}\}$  of $n>1$  mobile 
computational entities, called {\em robots}, operating in the
 Euclidean plane $\mathbb R^2$. The robots are {\em anonymous} (i.e., they are indistinguishable by their appearance), {\em autonomous} (i.e., without central control), {\em homogeneous} (i.e., the all execute the same program). Viewed as points they can move freely in the plane. Each robot is equipped with a local coordinate system (in which  it it is always at its origin), and it is able
to observe the positions of the other robots in its local coordinate system.
The robots  are {\em disoriented}; that is, there might not be consistency between the coordinate systems of different robots at the same time, or the same robot at different times\footnote{This is also called {\em variable} disorientation; restricted forms  (e.g.,  {\em static} disorientation, where each local coordinate system remains always the same) have been considered for these systems.}.
We assume that the robots however have {\em chirality}; that is, they 
agree on the  the same circular orientation of the plane (e.g., ``clockwise'' direction). 

At any time, a robot is either {\em active} or {\em inactive}.  When active, a robot $r$ executes a $\mathit{ Look}$-$\mathit{Compute}$-$\mathit{Move}$ ($\mathit{LCM}$) cycle.
Each cycle is compose of three  operations:
\begin{enumerate}
\item {\em Look:} The robot obtains an instantaneous snapshot of the positions occupied by the other robots (expressed in its own coordinate
system)\footnote{This is called the {\em full visibility} (or unlimited visibility)  setting; restricted forms of visibility have also been considered for these systems~\cite{FPSW05}.}. 
We do not assume that the robots are capable of strong multiplicity detection~\cite{FPS}.
\item {\em Compute:} The robot executes its algorithm using the snapshot as input. The result of the computation is a destination point.
\item {\em Move:} The robot moves  to the computed destination\footnote{This is called the {\em rigid mobility} setting;  restricted forms of mobility 
(e.g., when movement may be interrupted by an adversary), called {\em non-rigid mobility}  have also been considered for these systems.}; 
if the destination is the current location, the robot stays still and the move is said to be null.
 \end{enumerate}
\noindent After executing a cycle,  a robot becomes inactive. All robots are initially inactive. The time it takes to complete a cycle is assumed to be finite and the operations $\mathit{Look}$
 and $\mathit{Compute}$  are  assumed to be instantaneous.

In the standard model, $\OB$, the robots are also {\em silent}: 
they have no explicit means of communication; 
furthermore, they are {\em oblivious}: at the start of a cycle, a robot has no
memory of observations and computations performed in previous cycles.

In the other common model, $\LUM$, 
each robot $r$ is equipped with
a persistent register $Light[r]$, called {\em light}, whose value  
 called {\em color}, is from a constant-sized set $C$ and 
 is visible by the robots.
The color of the light can be set in each cycle by $r$ at the end of
its {\em Compute} operation, and is not automatically reset at the end of a cycle.
In $\LUM$, the {\em Look} operation produces a colored snapshot; i.e., it returns the set of pairs 
 $(position,color)$ of the other robots.
It is sometimes convenient to describe a robot $r$ as having $k\geq 1$ lights, denoted
 $r.light_1,\ldots,r.light_k$, where the values of $r.light_i$ are from a finite set of colors
 $C_i$, 
 and to consider $Light[r]$ as a $k$-tuple of variables; clearly, this corresponds to $r$ having a
 single light that uses $\Pi^{k}_{i=1}|C_i|$ colors.
 Note that if $|C|=1$, this case corresponds to the $\OB$ model.

 Two submodels of $\LU$, $\FS$ and $\FC$, have been defined and investigated, each offering only one
 of its two capabilities, persistent memory and 
direct means of communication, respectively.  
 In  $\FS$,  
  a robot can only see the color of its own light;   
 thus, the color merely encodes an internal state.
 Therefore, robots are {\em silent}, as in $\OB$, but they  
  are {\em finite-state}. 
In $\FC$, a robot can only see the color of the light of the other robots;
thus,  a robot can communicate to the other robots the color of its light 
but does not remember its own state (color).
 Thus,  robots are enabled with {\em  finite-communication}  but are {\em oblivious}. 
 
 In all the above models,
 a {\em configuration} ${\cal C}(\T)$ at time $\T$ is the multiset of the $n$ pairs 
 $(r_i(\T), c_i(\T))$, where $c_i(\T)$ is the color of robot $r_i$ at time $\T$.

\subsection{Schedulers,  Events}
With respect to the activation schedule of the robots, and the duration of their $\mathit{LCM}$  cycles, the fundamental distinction is between the {\em synchronous} and {\em asynchronous} settings.

In the {\em synchronous} setting (\SSY), also called {\em semi-synchronous} and first studied in \cite{SY},  time is divided into discrete
intervals, called {\em rounds}; in each round, a non-empty set of robots is activated  and  they simultaneously perform
a single  $\LK$-$\CP$-$\M$ cycle in perfect synchronization. 
The selection of which robots are 
activated at a given round is made by an adversarial scheduler, constrained only 
to be fair
(i.e., every robot is activated infinitely often). 
The particular  synchronous setting, where every robot is  activated in every round
is called {\em fully-synchronous} (\FSY).
In a synchronous setting, without loss of generality,  the expressions ``$i$-th round" and ``time $t=i$" are used as synonyms.

In the {\em asynchronous} setting (\ASY), first studied in \cite{FPSW99}, 
 there is no common notion of time,  the duration of each phase is finite 
 but unpredictable and might be different in different cycles, 
 and each robot is activated 
 independently of the others.
The duration of the phases of each cycle as well as the decision of when 
 a robot is activated is
controlled by an adversarial scheduler, constrained only to be fair,
i.e., every robot must be activated infinitely often.

In the asynchronous settings, the execution  by a robot of any of the operations  $\mathit{Look}$, $\mathit{Compute}$ and $\mathit{Move}$  is called an {\em event}.
We associate relevant time information to events:
for the  $\mathit{Look}$  (resp., $\mathit{Compute}$) operation,
which is instantaneous, the relevant time is $\T_L$  (resp., $\T_C$) when the event occurs;
 for the    $\mathit{Move}$ operation,   these are 
the times $\T_B$  and $\T_E$ when the event begins and ends, respectively.
Let ${\cal T} =\{\T_1, \T_2, ...\}$ denote the infinite ordered set of all relevant times; i.e.,
$\T_i < \T_{i+1}, i\in \N$.  In the following, to simplify the presentation and
without any loss of generality, we will refer to $\T_i\in {\cal T}$  simply
by its index $i$; i.e.,  the expression ``time $t$'' will be used to mean
``time $\T_t$''.

 In our analysis of \ASY, we will
 also consider and make use of  the following submodels of  \ASY,
 defined by the level of atomicity of the  $\LK$, $\CP$ and $\M$ 
 operations.

\begin{itemize}
\item \textbf{$\LC$-atomic-\ASY}:
The scheduler does not allow any robot $r$ to
perform a $\LK$ operation while another robot $r'\neq r$
is performing its 
$\CP$ 
operation in that cycle ~\cite{DKKOW19,OWD}.
Thus,  in the $\LC$-atomic-\ASY\ model,
it can be  assumed that, in every cycle, the
$\mathit{Look}$ and $\mathit{Comp}$ operations  
are performed simultaneously and atomically
and that $t_L=t_C$.

\item \textbf{$M$-atomic-\ASY}:
The scheduler does not allow any robot $r$ to
perform a $\LK$ operation while another robot $r'\neq r$
is performing its 
$\M$ 
operation in that cycle  \cite{DKKOW19,OWD}.
In this case, $\M$ operations 
 (called $M$-operations) 
in all cycles can be considered to be performed instantaneously
and that $t_B=t_E$.

\item \textbf{$\CM$-atomic-\ASY}:
The scheduler does not allow any robot $r$ to
perform a $\LK$ operation while another robot $r'\neq r$
is performing a $\CP$ or
$\M$ 
operation in that cycle. 
Thus, in this model, in every cycle  the operations $\CP$ and $\M$,  denoted as CPM,
can be considered as performed simultaneously and atomically, 
and $t_C=t_B=t_E$.
\end{itemize}


To complete the description, two additional specifications are necessary.\\
 Specification 1.
In presence of visible external lights (i.e., models $\LU$ and  $\FC$),
if a  robot $r$  changes its color  in the $\mathit{Comp}$  operation
at time $t\in {\cal T}$, by definition, its new color will become visible only at time $t+1$.


Specification 2. 
Under the \textbf{$M$-atomic-\ASY} and \textbf{$\CM$-atomic-\ASY}
schedulers,
if a robot $r$ ends a non-null $\mathit{Move}$ operation
at time $t\in {\cal T}$, by definition, its new position will become visible only at time $t+1$.

Note that, the model where  the $\LK$, $\CP$, and $\M$ operations are considered 
as a single instantaneous atomic operation
(thus referable to as  \textbf{$\LCM$-atomic-\ASY})  is 
obviously equivalent to \SSY. 

 In the following, for simplicity of notation, we shall use the symbols 
 ${F}$, ${S}$, ${A}$, ${A_{LC}}$, ${A_M}$, and ${A_{CM}}$ to 
 denote the schedulers \FSY, 
\SSY, \ASY, $\LC$-{\bf atomic}-\ASY, $M$-{\bf atomic}-\ASY, and $\CM$-{\bf atomic}-\ASY, respectively.

\subsection{Problems and Computational Relationships}
 Let ${\cal M} = \{\LU, \FC,\FS,\OB \}$  be the set of  models under investigation and 
${\cal S}= \{ F, S, A, A_{LC}, A_{M}, A_{CM}\}$ be the set of  schedulers 
under consideration.

A problem to be solved (or task to be performed) is described by a set of 
{\em temporal geometric predicates},
 which implicitly define the {\em valid}  initial, intermediate, and (if existing) 
terminal\footnote{A terminal configuration is one in which, once reached, the robots no longer move.} configurations, 
as well as restrictions (if any) on the size  $n$ of the set $R$ of robots.

An algorithm ${\cal A}$ {\em solves} a problem $P$ in model $M \in {\cal M}$ under scheduler 
$K\in  {\cal S}$
 if, 
starting from any valid initial configuration, 
any execution by $R$ of ${\cal A}$ in $M$ under $K$
 satisfies the  temporal geometric  predicates of $P$.

Given a model $M \in {\cal M}$ and  a scheduler $K\in  {\cal S}$, we denote by
$M^K$,
 the set of problems solvable 
by robots in $M$ 
under adversarial scheduler $K$.
Let $M_1, M_2\in{\cal M}$ and $K_1, K_2\in{\cal S}$.
\begin{itemize} 
\item We say that  model $M_1$  under scheduler  $K_1$
is {\em computationally not less powerful than} 
model $M_2$ under $K_2$, denoted by
$M_{1}^{K_1} \geq M_{2}^{K_2}$,
 if $M_{1}(K_1) \supseteq M_{2}(K_2)$.

\item We say that 
 $M_1$  under  $K_1$
is {\em computationally more powerful than} 
$M_2$ under $K_2$,  
denoted by
$M_{1}^{K_1} >  M_{2}^{K_2}$,
if 
$M_{1}^{K_1} \geq M_{2}^{K_2}$ 
and
$(M_{1}(K_1) \setminus M_{2}(K_2))  \neq \emptyset$.

\item We say that  $M_1$  under  $K_1$
and  $M_2$  under  $K_2$, are {\em computationally equivalent }, denoted by  
$M_{1}^{K_1} \equiv  M_{2}^{K_2}$,
if $M_{1}^{K_1} \geq M_{2}^{K_2}$ and $M_{2}^{K_2} \geq M_{1}^{K_1}$.

\item Finally, we say that 
$K_1$
$K_2$, are {\em computationally orthogonal} (or {\em incomparable}), denoted by  
$M_{1}^{K_1} \bot  M_{2}^{K_2}$, 
if 
 $(M_{1}(K_1) \setminus M_{2}(K_2))  \neq \emptyset$
 and 
 $(M_{2}(K_2) \setminus M_{1}(K_1))  \neq \emptyset$.

 \end{itemize}

 Trivially, 
 
 \begin{lemma}
 For any $M\in{\cal M}$ and any  $K\in{\cal S}$:
 \begin{enumerate}
 \item $M^{F} \geq M^{S} \geq M^{A_{LC}} \geq M^{A}$  
 \item $M^{F} \geq M^{S} \geq M^{A_{CM}} \geq M^{A_{M}} \geq M^{A}$
  \item $\LUM^{K} \geq \FS^{K} \geq \OB^{K}$
\item $\LUM^{K} \geq \FC^{K} \geq \OB^{K}$
\end{enumerate}
\end{lemma}

Let us also recall the following equivalence established in \cite{DFPSY}:
 \begin{lemma} [\cite{DFPSY}] $\LU^{A} \equiv \LU^{S}$
\end{lemma}
\noindent that is, in the $\LU$ model,  there is no computational difference  
between \ASY\ and \SSY.

Observe that, in all models, any restriction of the adversarial power of the asynchronous scheduler
does not decrease (and possibly increases) the computational capabilities
of the robots in that model. 
In other words,  if $A_\alpha$  is a restricted
scheduler of $A_\beta$, then  $M^{A_\beta} \leq M^{A_\alpha}$
 for any robot model $M \in {\cal M}$.

Note that the difference between $A_{CM}$ and $A_{M}$ is that there exists just one type of configuration that can be observed in $A_{M}$ but cannot be observed in $A_{CM}$:  the one before moving but after computing. 
As for $X \in \{\FS,\OB \}$, since  robots cannot observe the colors of the other robots, we have $X^{A_{CM}} \equiv X^{A_{M}}$ 
and  
$X^{A_{LC}} \equiv X^{A}$. 

\section{The \texorpdfstring{$\OB$}{OB}  Computational Landscape}
\label{sec-Oblot}

\subsection{Separating \SSY\ from \ASY}

In this section we prove that, under \SSY,  the
robots  in  $\OB$ are strictly more powerful than
under  ${\cal A}_{M}$, thus separating
\SSY\ from \ASY\  in $\OB$.



To do so, we consider the classical {\tt Collisionless Line Convergence}  (\CLCv)
problem, where two robots, r and q, must converge to a 
common location,  moving on the line connecting them, without ever  crossing each other; i.e., \CLCv\ is defined by the predicate 
\begin{align}
\begin{split}
CLC \equiv \Big[
&\{\exists \ell \in \mathbb{R}^2, \forall\epsilon\geq 0, \exists T\geq 0, \forall t \ge T:  |r(t)-\ell|+ |q(t)-\ell| \le \epsilon\}, \\
&{\textbf{and}}\  \{\forall t\geq 0: r(t),q(t)\in\overline{r(0)q(0)}\}, \\
&{\textbf{and}}\ \{\forall t\geq 0: 
dis(r(0),r(t))\leq dis(r(0),q(t)), 
dis(q(0),q(t))\leq dis(q(0),r(t))\} \Big]
\end{split}
\nonumber
\end{align}
and we focus on the monotone version 
of this problem defined below.

%
\begin{definition}
\label{def:MLCv}
{\bf MONOTONE LINE CONVERGENCE} $($\MLCv$)$
 The two robots, $r$ and $q$ must  solve the {\tt Collisionless Line Convergence} problem without 
 ever increasing the distance between them.
\end{definition}
\noindent In other words, an algorithm solves
\MLCv\   iff it
satisfies the following predicate:
$$MLC \equiv [ CLC\  {\bf and}\ \{\forall t'\geq t, |r(t')-q(t')| \le |r(t)-q(t)|\}]$$


 First observe that  
\MLCv\ can   be solved in $\OB^S$.

\begin{lemma}\label{lem:CNV-WOC1}
 \MLCv$\in \OB^{S}$. This holds even under  
non-rigid movement and in absence  of chirality.
\end{lemma}

\begin{proof}

It is rather immediate to see that the simple  protocol using the strategy
``move to half distance" satisfies 
the MLC predicate and thus solves the problem.
\end{proof}

On the other hand, \MLCv\ 
 is not solvable in $\OB^{A_{M}}$.

\begin{lemma}\label{lem:CNV-WOC2} 
  \MLCv $\not \in  \OB^{A_{M}}$ even under fixed disorientation 
 and agreement on the unit of distance.   
\end{lemma}
\begin{proof}

By contradiction, assume that there exists an algorithm ${\cal A}$ 
that solves \MLCv\ in $\OB^{A_{M}}$. 
Let  the two robots, $r$ and $q$, have the same unit of distance,  initially each
see the other on the positive direction of the $X$ axis
and their local coordinate system  not change during the execution of ${\cal A}$.
Three observations are in order.

(1) First observe that, by  the predicates defining  \MLCv, if a robot moves, it must move towards the other,
and in this particular setting, it must stay on its $X$ axis.

(2) Next observe that, every time a robot is activated and executes ${\cal A}$,
it must move. In fact, if, on the contrary,   ${\cal A}$  prescribes that a robot activated
at some distance $d$ from the other must not move,
then, in a fully synchronous execution of ${\cal A}$
where both robots are initially at distance $d$, neither
of them will ever move and, thus, will never 
converge. 

(3) Finally observe that, when  robot $r$ moves towards $q$ on the $X$ axis
after seeing it  at distance $d$,  the length $f(d)$ of the computed move 
is the same as that $q$ would  compute if seeing $r$  at distance $d$.

Consider now  the following execution ${\cal E}$ under ${A_{M}}$:
 Initially both robots  are simultaneously activated, and are at
 distance $d$ from each other. 
Robot $r$ completes its computation and executes 
the move instantaneously (recall, they are operating under
$A_{M}$), and continues to be activated and to execute 
${\cal A}$ while robot $q$ is still in its initial 
computation. 

Each  move by $r$  clearly reduces the distance
between the two robots. More precisely, by observation (3), after $k\geq 1$ moves,
the distance will be reduced from $d$ to $d_{k}$ where  $d_0=d$ and 
$d_{k>0}  =d_{k-1} - f(d_{k-1}) = d - \sum_{0 \leq i < k} f(d_i)$.\\

{\bf Claim.} {\em After a finite number of moves of $r$,
 the distance between the two robots becomes smaller that $f(d)$}.

{\bf Proof of Claim.}
By contradiction, let $r$ never get closer than $f(d)$ to $q$; that is
for every $k>0$, $d_{k}> f(d)$.

Consider then the execution $\hat{{\cal E}}$ of ${\cal A}$ under the \RR\ synchronous scheduler:
the robots, initially at distance $d$, are activated one per round, at alternate rounds.
Observe that, since ${\cal A}$ is assumed to be correct under ${A_{M}}$,
it must be correct also under \RR. This means that, 
starting from the initial distance $d$, for any fixed distance $d' >0$, the two robots become closer 
than $d'$. Let $m(d')$ denote the number of rounds for this to occur; then, 
  the distance between 
them becomes smaller than $f(d)$ after $m(f(d))$ rounds. 
Further observe that, after round $i$, the distance $d_i$ between 
them is reduced by $f(d_i)$. 
Summarizing, 
$d_{m(f(d))}  =  d - \sum_{0 \leq i < m(f(d))} f(d_i) < f(d)$, 
contradicting that  $d_{k}> f(d)$ for every $k>0$.
\qed\\

Consider now  the execution ${\cal E}$ at the time  
the distance becomes smaller that $f(d)$;
 let robot $q$ complete its computation at that time and
perform its move, of length $f(d)$, towards $r$. This move then
creates a collision, contradicting the correctness 
of {\textit A}.
\end{proof}

From Lemmas \ref{lem:CNV-WOC1} and \ref{lem:CNV-WOC2},
and since  $\OB^{A_M} \geq \OB^{A}$ by definition, 
the main result now follows:

\begin{theorem}\label{th:domFSOBSoverFSOBAM} 
 $\OB^{S} >$ $\OB^{A}$ 
\end{theorem}

In other words,  
under the synchronous scheduler \SSY, $\OB$ robots are 
strictly more powerful than when under the asynchronous scheduler \ASY.
This results provides a definite positive answer to the
long-open question of whether there exists a
computational difference between synchrony and asynchrony in 
$\OB$.

\subsection{Refining the \texorpdfstring{$\OB$}{OB}  Landscape}
\label{sec: OBlandscape}
We can refine the $\OB$ landscape as follows;
By definition, $\OB^{A_M} \geq \OB^{A}$. 
Consider now the following problem for $n=4$ robots.


\begin{figure}[tbh]
    \centering
\centering\includegraphics[keepaspectratio, width=0.54\textwidth]{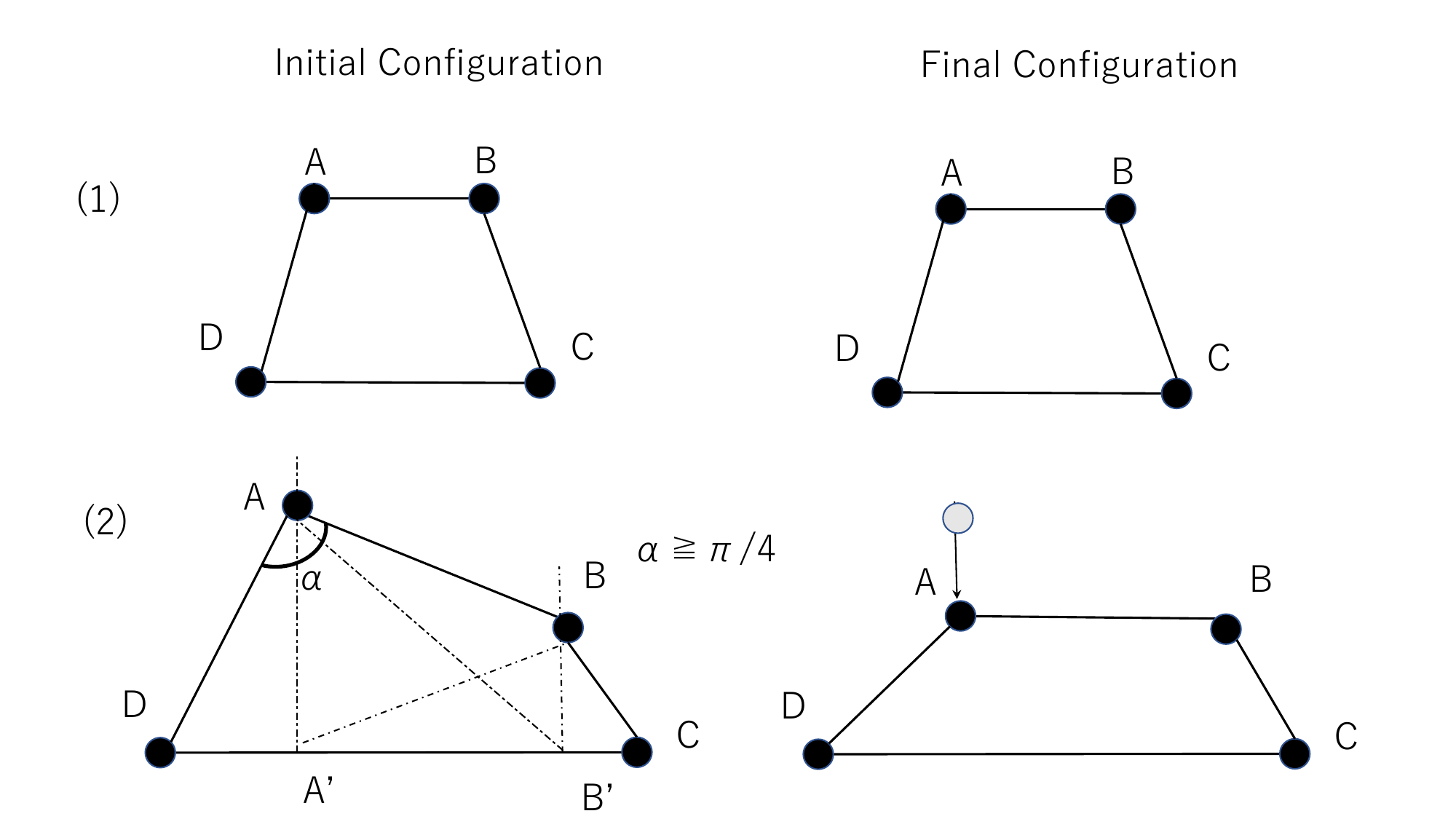}
\centering\includegraphics[keepaspectratio, width=0.44\textwidth]{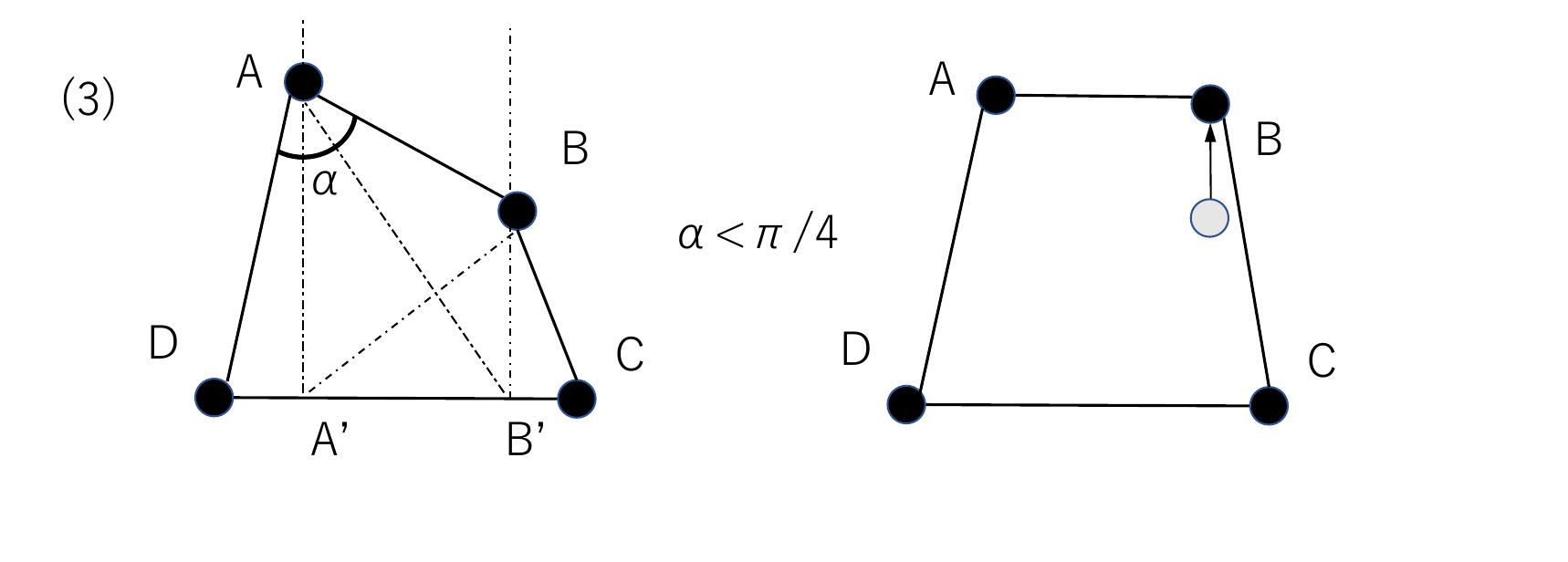}
    \caption{TRAPEZOID FORMATION (TF).
    }
    \label{fig:TF}
\end{figure}

\begin{definition}
{\bf  TRAPEZOID FORMATION  (TF)} :  
Consider a set of four robots, $R=\{a,b, c, d\}$ whose initial configuration  forms a
convex quadrilateral $Q=(ABCD)= (a(0)b(0)c(0)d(0))$ with one side, say $\overline{CD}$, longer than all others. The task is to
transform $Q$ into a trapezoid $T$,
subject to the following conditions:\\
(1)  If $Q$ is a trapezoid,
the configuration must stay unchanged (Figure~\ref{fig:TF}(1)); i.e.,
$$TF1 \equiv
[\ Trapezoid(ABCD) \Rightarrow 
\{ \forall t > 0, r\in\{a,b,c,d\}  : r(t)= r(0) \}\ ]$$
(2) Otherwise,
without loss of generality,
 let $A$ be farther than $B$ from $CD$.
  Let $Y(A)$  (resp., $Y(B)$) denote the perpendicular lines 
from $A$ (resp., $B$) to $CD$ meeting $CD$ in $A'$ (resp. $B'$), and
  let $\alpha$ be the smallest angle between 
$\angle{BAA'}$ and $\angle{ABB'}$.\\
(2.1)
If $\alpha \geq \pi/4$ then the robots must form the 
trapezoid
shown in Figure~\ref{fig:TF}(2), 
where the location  of $a$ 
is a translation of its initial one on the line $Y(A)$, and that 
of all other robots is  unchanged; specifically,
$$TF2.1 \equiv [\ (\alpha \geq \pi/4)\ 
\Rightarrow  
 \{ \forall t\geq0,  r\in\{b,c,d\}  : r(t)= r(0), a(t) \in Y(A) \}\ {\bf and}\ $$
 $$ \{
\exists t>0 : \forall t'\geq t,\ 
\{ \overline{a(t')b(t')}\ || \overline{CD}\}\
{\bf and}\  
\{ a(t')= a(t)\}\ 
\}\ ]
$$
 (2.2) If instead $\alpha < \pi/4$
then the robots must form the 
trapezoid
shown in Fig.~\ref{fig:TF}(3), 
where the location of all robots but $b$  is  unchanged, and that of $b$ 
is a translation of its initial one on the line $Y(B)$; specifically,
$$TF2.2 \equiv [\ (\alpha < \pi/4)\ 
\Rightarrow  
 \{ \forall t\geq0,  r\in\{a,c,d\}  : r(t)= r(0), b(t) \in Y(B) \}\ {\bf and}\ $$
 $$ \{
\exists t>0 : \forall t'\geq t,\ 
\{ \overline{a(t')b(t')}\ || \overline{CD}\}\
{\bf and}\  
\{ b(t')= b(t)\}\ 
\}\ ]
$$
\end{definition}

\indent Observe that  
$TF$  can  be solved in $\OB^{A_{M}}$.

\begin{lemma}\label{lem:TF1}
$ \text{TF} \in \OB^{A_{M}}$, even in
 absence of chirality.
\end{lemma}

\begin{proof}

It is immediate to see that the following
simple set of rules solves TF in $\OB^{A_{M}}$. 

Rule 1: If the observed configuration 
is as shown in Figure~\ref{fig:TF} (1), 
the configuration is already a trapezoid, and
no robot performs any move ($TF1$).

Rule 2: Let the configuration  be as shown in Figure~\ref{fig:TF} (2). 
Whenever  observed by $b, c, d$, none of them moves;  
 when observed by $a$, $a$ moves to the desired point eventually creating a terminal configuration subject to Rule 1.
Since the scheduler is $\M$-atomic \ASY\, the other robots do not observe $a$ during this move, but only  after the move is completed.

Rule 3: Analogously, 
let the configuration  be as shown in Figure~\ref{fig:TF} (3).
Whenever  observed by  $a, c, d$, none of them moves;  when observed by $b$, 
$b$ moves to the desired point eventually creating a trapezoid and reaching a terminal configuration, unseen by all other robots
during this movement.
\end{proof}

 However,  
$TF$  cannot  be solved in $\OB^{A}$.  

\begin{lemma}\label{lem:TFoblot}
$\text{TF} \notin \OB^{A}$, even with fixed disorientation.
\end{lemma}

\begin{proof}

By contradiction, let ${\cal A}$ be an algorithm that always allows the four  $\OB$ robots to solve TF under the asynchronous  scheduler. Consider  the initial configuration
where  $a$ is further than $b$ from $\overline{CD}$, and   
$\alpha = \pi/4$. In this configuration, $a$ is  required to  move (along $Y(A)$) 
while no other robot is  allowed to move. Observe that, as soon as $a$ moves,
 it creates a configuration where $a$ is still  further than $b$ from $\overline{CD}$,  but  
  $\alpha' = \min \{ \angle{b(t)a(t)A'}, \angle{a(t)b(t)B'} \} < \pi/4$.
Consider now the execution of ${\cal A}$  in which
 $a$ is activated first,  and then $b$ is activated  while $a$ is moving; 
 in this execution, the configuration seen by $b$ requires  it to
to  move, violating  $TF2.1$ and contradicting the assumed correctness 
 of algorithm ${\cal A}$.

\end{proof}

Thus, by Lemmas \ref{lem:TF1} and \ref{lem:TFoblot},
we have  
\begin{theorem}
\label{thm:TFoblot}
\ \  $\OB^{A_M} > \OB^{A}$
\end{theorem}


\begin{theorem}
$\OB^{S}>\OB^{A_{M}}>\OB^{A_{LC}}\equiv\OB^{A}$
\end{theorem}
\begin{proof} (1) The equivalence $ \OB^{A_{LC}}\  \equiv \OB^{A}$ holds because, by definition, $\OB$ robots cannot distinguish between $A_{LC}$ and $A$; then, by 
Theorem \ref{thm:TFoblot}, $\OB^{A_{M}}\  > \OB^{A_{LC}}$.
(2) It follows  from Lemmas \ref{lem:CNV-WOC1} and \ref{lem:CNV-WOC2}. 
(3) It follows from (1) and Theorem \ref{th:domFSOBSoverFSOBAM}. \end{proof}

\section{The  \texorpdfstring{$\FC$}{FC} Computational Landscape}\label{sec:FCOM}

\subsection{Separating \SSY\ from \ASY\  in \texorpdfstring{$\FC$}{FC}}

We have seen (Theorem \ref{th:domFSOBSoverFSOBAM}) that,
to overcome the limitations imposed by asynchrony,
the robots must have some additional power with
respect to those held in $\OB$.

In this section, we show that  the communication capabilities of $\FC$ 
are not sufficient.
In fact, we prove that, under \SSY,  the
robots  in  $\FC$ are strictly more powerful than
under  ${\cal A}_{M}$, thus separating
\SSY\ from \ASY\  in $\FC$. 
To do so, we use the problem MLCv again. 

Observe that \MLCv\ can be solved even in $\OB^S$ (Lemma~\ref{lem:CNV-WOC1}), and thus
in $\FC^S$.

\begin{lemma}\label{lem:MLCv-pos}
$\MLCv \in \FC^{S}$; this holds even under variable disorientation, non-rigid movement and in absence  of chirality.
\end{lemma}

On the other hand, \MLCv\ 
 is not solvable in $\FC^{A_{M}}$.

\begin{lemma}\label{lem:MLCv-imp}
$\MLCv \not \in  \FC^{A_{M}}$.  
\end{lemma}

\begin{proof}
Let us consider two robots, $r$ and $q$, and show that the adversary can activate them in a way that exploits variable disorientation to cause them to violate the condition of $\MLCv$.

We consider the execution in which the adversary always forces the robots to perceive the distance between $r$ and $q$ as 1, which is equivalent to the current unit distance of $X$. We define a function $f(c,d)$ as the length of the move taken by a robot when it observes color $c$ of the other robot and the true distance between the two robots is $d$ in the last Look phase. Since the distance always appears as 1 to the robots, the value $F(c) = f(c,d)/d$ is independent of $d$. We denote the initial color of the robots as $c_0$ and assume that $F(c_0)>0$, which does not affect generality as the adversary can activate $r$ and $q$ multiple times until both robots have a color $c$ such that $F(c)>0$. Without loss of generality, we also assume that $F(c_0)\le 1/2$. If $F(c_0)> 1/2$, it follows that $r$ and $q$ pass each other when the adversary activates both robots at time step 0, violating the condition of $\MLCv$. Therefore, we assume $0 < F(c_0)< 1/2$ without loss of generality.

Starting from time step 0, the adversary refrains from activating $r$ and instead activates only $q$ to move $\lfloor \log_{1-F(c_0)} F(C_0) \rfloor +1$ times. Since $q$ always perceives $c_0$ as the color of $r$ during this period, the distance between $r$ and $q$ decreases by a factor of $(1-F(C_0))$ with each move of $q$. As a result, the distance between $r$ and $q$ becomes smaller than $F(C_0) d_0 = f(c_0,d_0)$ after this period, where $d_0$ is the initial distance between $r$ and $q$. The adversary then activates $r$ to perform its Move phase. $r$ moves a distance of $f(c_0, d_0)$ and overtakes $q$, thereby violating the condition.
\end{proof}

From Lemmas \ref{lem:MLCv-pos} and \ref{lem:MLCv-imp},
and since  $\FC^{A_M} \geq \FC^{A}$ by definition, 
the main result now follows:

\begin{theorem} 
\label{th:domFCSoverFCA} 
 $\FC^{S} >$ $\FC^{A}$ 
\end{theorem}

\subsection{Refining the \texorpdfstring{$\FC$}{FC} Landscape}
\label{RefineFCOM}


In this section, we complete the characterization of the asynchronous landscape of $\FC$ proving
$\FC^{A}\equiv 
\FC^{A_{CM}}< \FC^{A_{LC}} \equiv \FC^{S}$.
Specifically, we prove the following two theorems. 

\begin{theorem}
\label{th:FCACM=FCA}
 $\FC^{A_{CM}} \equiv \FC^{A}$. This holds in absence of chirality.
 \end{theorem}


\begin{theorem}
\label{th:FCALC=FCS}
$\FC^{S} \equiv\FC^{A_{LC}}$.
\end{theorem} 

We will prove Theorems \ref{th:FCACM=FCA} and \ref{th:FCALC=FCS} in Sections \ref{sec:FCACM=FCA} and \ref{sec:FCALC=FCS}, respectively.
By these two theorems and Theorem \ref{th:domFCSoverFCA}, we immediately obtain
the following separation:
\begin{theorem}
\label{th:FLC>FCM}
 $\FC^{A_{LC}} > \FC^{A_{CM}}$.
\end{theorem}

Note that, since $A_{CM} \leq A_{M} \leq A$, Theorem \ref{th:FCACM=FCA} implies the following corollary.

\begin{corollary}\label{co:FCACM=FCAM}
 $\FC^{A_{M}} \equiv \FC^{A}$. This holds in absence of chirality.
 \end{corollary}

\subsubsection{Proof of Theorem \ref{th:FCACM=FCA}}
\label{sec:FCACM=FCA}

In this section, we show hat
every problem solvable by a set of $\FC$ robots  under ${A_{CM}}$ can also be solved 
 under \ASY.  We do so constructively: we present a {\em simulation} algorithm 
for $\FC$ robots   that allows them to correctly execute in \ASY\ 
any protocol ${\cal A}$ given in input (i.e., all its executions under   \ASY\ are
equivalent to  some executions under ${A_{CM}}$).  
 The simulation algorithm (called 
 \textit{SIM})  makes each $\FC$ robot 
  execute   $\mathcal{A}$   infinitely often,
never violating the conditions of scheduler $CM$-{\bf atomic}-\ASY. 
To achieve this, {\em SIM} needs an activated robot to be able to retrieve 
some  information about its past (e.g., whether or not it has ``recently" executed
$\mathcal{A}$). Such information can  obviously be encoded and persistently stored by the robot in
the color of its own light; but, since an $\FC$ robot cannot see the color of its light, 
the robot cannot access the stored information. However, this information can be seen 
by the other robots, and hence can be communicated by some of them (via the color of their lights)  to the needing robot. This can be done efficiently as follows. Exploiting  chirality, the robots can agree
at any time on a circular ordering of the nodes where robots are located, so that for any such a location $x$
both its predecessor  $\CoP(x)$ and its successor  $\CoS(x)$ in the ordering are uniquely identified,
with  $\CoP(\CoS(x))=x$;  all robots located at  $x$ then become responsible for communicating the needed information to  the robots located at $\CoS(x)$\footnote{Although we use chirality to determine the cyclic order, this assumption can be circumvented by slightly increasing the number of light colors and deciding the color of the corresponding robot using local 'suc' and 'pred'~\cite{FPS19}.}.

\Newcodeline
\begin{algorithm}[t]
\caption{{\tt SIM(A)}:  predicates and subroutines for robot $r$ at location $x$} 

\label{algo:simA}
{\small
\begin{tabbing}
111 \= 11 \= 11 \= 11 \= 11 \= 11 \= 11 \= \kill
{\em Assumptions}: Let $x_0, x_1, \ldots, x_{m-1}$ be the circular arrangement on the configuration $C$ ($m \geq 2$), \crm
and let define $\CoS(x_i)=x_{i+1 \mod m}$ and $\CoP(x_i)= x_{i-1 \mod m}$.\crm\crm
%
{\em State Look}\crm
Observe, in particular,   
  $\CoP(x).state(\CoS.state)$, $\CoS(x).state(\CoS.state)$, and   $\rho.phase(\rho \neq r)$; \\
  as well as $r.state.here$ (the set of states seen by $r$ at its own location $x$. \\Note that, for this, $r$ cannot see its own  color).\crm\crm
\color{black}

{\em predicate}  {\em Is-all-phases(p:phase)} \crm
$\forall \rho(\neq r )(\rho.phase = p)$ \crm
\crm



{\em predicate}  {\em Is-phases-mixed(p, q: phase)} \crm
$\forall \rho(\neq r )[(\rho.phase =p)$  {\bf or} $(\rho.phase =q)$] \\ {\bf and} [{\bf not} {\em Is-all-phases(p)}] {\bf and} [{\bf not}  {\em Is-all-phases(q)}]\crm
\crm
{\em predicate}  {\em Is-exist-M}\crm
$\exists \rho(\neq r )[(\rho.state=M)$ {\bf or} $(M \in \rho.suc.state)]$\crm
\crm
{\em predicate}  {\em Is-all(s: state)}\crm
$\forall \rho(\rho.state=s)$ {\bf and} $own.state=\{s\}$\crm
\crm
{\em function}  {\em r.own.state}: set of states\crm
 \>$own.state \leftarrow  \CoP(x).\CoS.state- x.state.here$,\crm
 \>where $x.state.here$ corresponds
to  the set of states seen by $r$ at its own location $x$\crm 
\crm

{\em subroutine} {\em Copy-States-of-Neighbors}\crm
 \>r.\CoS.state $\leftarrow \CoS(x).state$\crm
\crm


{\em subroutine Reset-state-and-neighbor-state} \crm
\>\> $r.state \leftarrow W$\crm
\>\> $r.\CoS.state \leftarrow \{W\}$\crm
\crm
\end{tabbing}

}
\end{algorithm}

\Newcodeline
\begin{algorithm}
\caption{{\tt SIM(A)} - for robot $r$ at location $x$}
\label{algo:simAmain}
{\small
\begin{tabbing}
111 \= 11 \= 11 \= 11 \= 11 \= 11 \= 11 \= \kill
%

{\em State Compute}\crm
\Cl \> r.des $\leftarrow$ r.pos  \crm
\Cl \> {\bf if} {\em Is-all-phases(1)} {\bf then} \crm
\Cl \>\>$Copy$-$state$-$of$-$Neighbors$\crm
\Cl \>\>$r.phase \leftarrow  1$\crm
\Cl \>\>{\bf if} {\em Is-all(F)} {\bf then}  $r.phase \leftarrow m$\crm
\Cl \>\>{\bf else if} {\em ($\exists \rho(\neq r )[(\rho.state=M)$)} {\bf then} $r.phase \leftarrow 2$\crm

\Cl \>\>{\bf else if} ({\em r.own.state=} $\{W\}$)  {\bf then}\crm
\Cl \>\>\>\> Execute the {\tt Compute} of  $\mathcal{A}$  //   determining my color $r.light$ and destination $r.des  $ //  \crm
\Cl \>\>\>\>$r.state \leftarrow M$\crm

\Cl \> {\bf else if} {\em Is-all-phases(2)} {\bf then} \crm%
\Cl \>\>$r.phase \leftarrow  3$\crm
\Cl \>\>$Copy$-$state$-$of$-$Neighbors$ \crm

\Cl \> {\bf else if} {\em Is-all-phases(3)} {\bf then}\crm
\Cl \>\>$Copy$-$state$-$of$-$Neighbors$\crm
\Cl \>\>$r.phase \leftarrow  3$\crm
\Cl \>\> {\bf if} {\em Is-exist-M}  {\bf then}\crm

\Cl \>\>\>{\bf if} {\em r.own.state=} $\{M\}$ {\bf then}\crm
\Cl \>\>\>\>$r.state \leftarrow F$\crm
\Cl \>\>\>$Copy$-$state$-$of$-$Neighbors$ \crm
\Cl \>\> {\bf else}// $no$-$M$//\crm
\Cl \>\>\>$r.phase \leftarrow  1$\crm
\Cl \>\>\>$Copy$-$state$-$of$-$Neighbors$ \crm
\Cl \> {\bf else if} {\em Is-phase-mixed(1,2)} {\bf then}\crm
\Cl \>\>\>$r.phase \leftarrow  2$\crm
\Cl \> {\bf else if} {\em Is-phase-mixed(2,3)} {\bf then}\crm
\Cl \>\>\>$r.phase \leftarrow  3$\crm
\Cl \>\>\>$Copy$-$state$-$of$-$Neighbors$ \crm

\Cl \> {\bf else if} {\em Is-phase-mixed(1,3)} {\bf then}\crm
\Cl \>\> \>$r.phase \leftarrow  1$ \crm
\Cl \>\>\>$Copy$-$state$-$of$-$Neighbors$ \crm

\Cl \> {\bf else if} {\em Is-all-phases(m)} {\bf then} //Reset state//\crm
\Cl \> \>{\em Reset-state-and-neighbor-state} \crm
\Cl \>\> {\bf if} $\exists \rho \neq r (\rho.state=F)$ {\bf then}\crm
\Cl \> \> \>r.phase $\leftarrow m$  \crm
\Cl \>\> {\bf else} //There does not exist $F$//\crm
\Cl \> \> \>r.phase $\leftarrow 1$ \crm

\Cl \> \label{line:1mmixed}{\bf else if}  {\em Is-phase-mixed(1,m)}  {\bf and}   {\em Is-all(F)}  {\bf then} $r.phase \leftarrow m$ \crm
\Cl \> {\bf else if}  {\em Is-phase-mixed(1,m)}  {\bf and}   {\em Is-all(W)}  {\bf then} $r.phase \leftarrow 1$ \crm

\crm

{\em State Move}\crm
Move to $r.des$;\crm
\end{tabbing}
}
\end{algorithm}

 Let $\mathcal{A}$ be an algorithm for  $\FC$ robots 
 in $CM$-{\bf atomic}-\ASY,
 and let  $\mathcal{A}$  use a light of $\ell$ colors: $C=\{c_0, c_1, \ldots c_{\ell-1} \}$. 
  It is assumed that, in any initial configuration ${\mathcal C}$,
   the number of distinct locations\footnote{  In $\FC$, 
 by definition, if all robots of the same color are located on the same position, they would not be able to see anything including themselves and they could not perform any task.} is $m \geq 2$. 
 
The pseudo code of the simulation algorithm is presented in {\bf Algorithm~1} (predicates and subroutines) and {\bf Algorithm~2} (main program). 
 
The simulation algorithm is composed of four phases. 
To execute the simulation algorithm, a robot $r$ uses four externally visible persistent lights:
\begin{enumerate}
\item $r.light$  $\in C$, indicating its own light used in the execution of $\mathcal{A}$; initially, $r.light=c_0$;
\item $r.phase$ $\in \{1,2,3,m\}$, indicating the current
phase of the simulation algorithm;
  initially $r.phase=  1$;
\item $r.state \in \{W, M, F\}$, indicating the state of $r$ in its execution of the simulation; 
initially,  $r.state=W$;  
\item $r.\CoS.state\in 2^{\{W,M,F\}}$, indicating the  set of states  at \CoS(x), where $x$ is the current location of $r$; initially,  r.\CoS.state=$\{W\}$. 
\end{enumerate}

\noindent Summarizing, each robot $r$ has $Light[r]=\langle r.light, r.phase, r.state, r.\CoS.state \rangle$. For a location $x$ and  $l\in\{light,phase,state,\CoS.state\}$, let $x.l= \cup_{\text{r at x}}r.l$ 
denote the set of the lights $r.l$ of the robots at location $x$.

Figure~\ref{fig:transition-sim-SSY-by-ASY-LC-in-FCOM} shows the transition diagram as the robots change phase's values. 
Informally, the simulation algorithm is composed of three main phases which are continuously repeated
and a fourth one which is occasionally performed.
Each execution of the three main phases corresponds to a single execution of  $\mathcal{A}$, each 
satisfying the $CM$-{\bf atomic} condition,  by some robots. The three main phases are repeated until 
every robot has executed $\mathcal{A}$ at least once, ensuring fairness. Appropriate flags are set up to detect
 this occurrence; a  "mega-cycle"  is said to be completed, and after the execution of the fourth phase
 (a reset), a new mega-cycle is started (continuing
 the simulation of the execution of $\mathcal{A}$ through the continuing execution of  the three phases).
 In other words, in each mega-cycle all robots are activated and execute\footnote{In each phase of mega-cycles, at most one robot may execute the simulated algorithm more than once.} ${\cal A}$ under the $CM$-{\bf atomic} condition.

\begin{figure}
\centering\includegraphics[keepaspectratio, width=14cm]{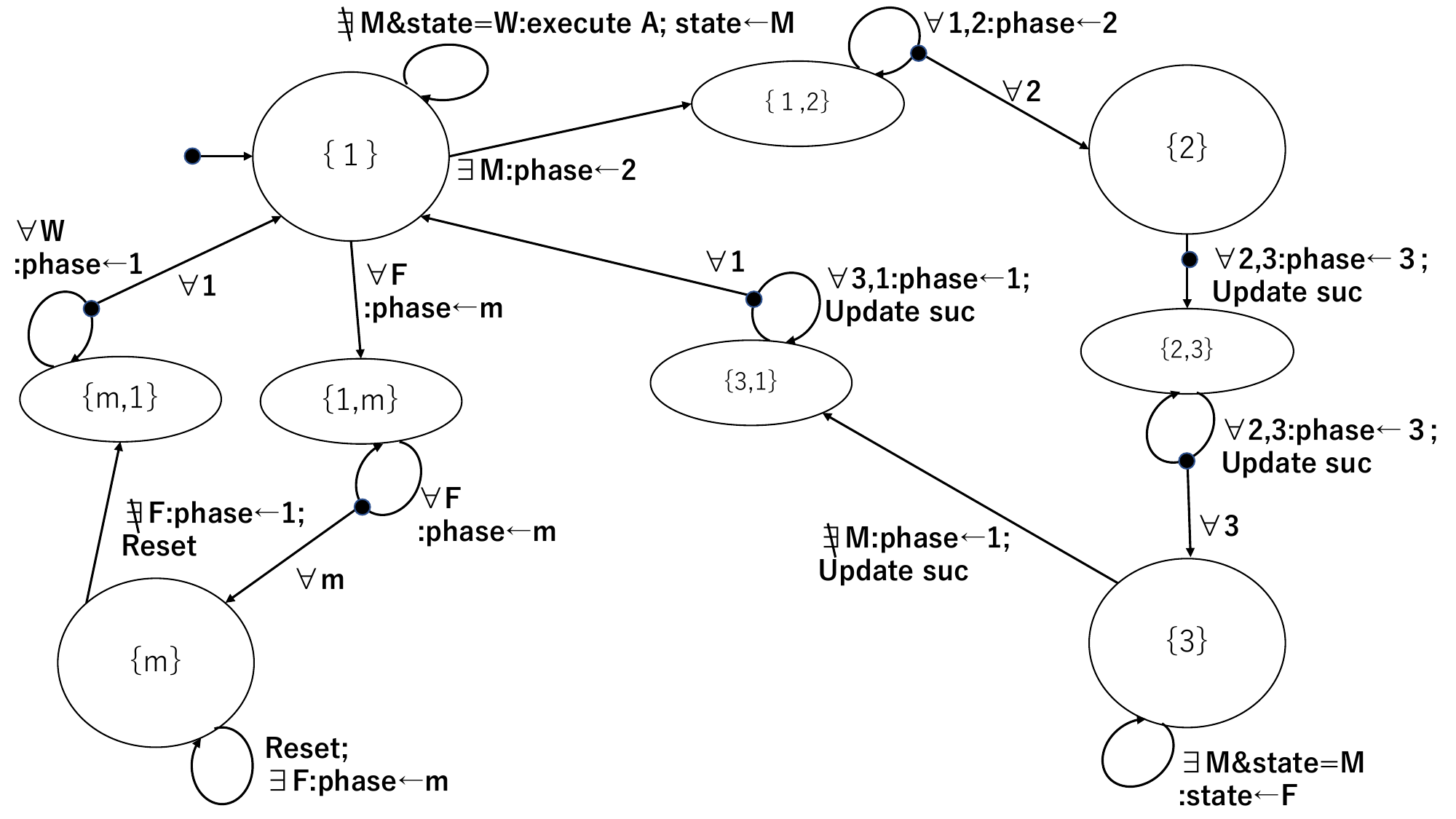}
    \caption{Transition Diagram of {\tt SIM(A)}. In the figure,label $condition:action$ on each arrow means if an activated robot observes $condition$, the robot performs $action$ and the configuration is changed to one the head of the arrow points. Arrows with branches indicate that if the condition on the arrow is satisfied, the transition occurs. {\bf Update suc} and {\bf Reset} indicate that "update suc.state" and "reset state and suc.state," respectively. Some arrows are omitted because the figure gets complicated. In the cases of $(\alpha, \beta)=(1,2), (2,3), (3,1), (1,m), (m,1)$, there is direct transition from $\{ \alpha\}$ to  $\{\beta \}$ without via $\{ \alpha,\beta\}$}
    \label{fig:transition-sim-SSY-by-ASY-LC-in-FCOM}
\end{figure}

The details of these phases and between phases are as follows:


\begin{itemize}
\item {\bf Between Phases} ($\forall \rho [\rho.phase =p$ or $\rho.phase=q]$).
These are the transit states from Phase~$p$ to Phase~$q$. Activated robots only change $phase$ flags and do not execute the simulated algorithm nor change $state$ flags (see Fig.~\ref{fig:transition-sim-SSY-by-ASY-LC-in-FCOM}). Also from $2$ to $3$, and from $3$ to $1$, each robot updates its $\CoS.state$ during this mixed phases.

\item {\bf Phase 1-Perform Simulation} ($\forall \rho(\neq r) [\rho.phase =1]$). 

In the $\LK$-operation, 
$r$ understands to be in Phase~$1$ by detecting  $\rho.phase=1$ for any other robot $\rho$. 
After the $\LK$-operation,  each robot $r$ at $x$ can recognize its own
$r.state$ by using the predecessor's $\CoS.state$ and $r.state.here$ corresponding to the set of states seen by $r$ at its own location $x$ (Fig.~\ref{fig:determine-own-state}). Since we assume the agreement of chirality, the relation of $suc$ and $pred$ is uniquely determined\footnote{Although we use chirality to determine the cyclic order, this assumption can be circumvented by slightly increasing the number of light colors and deciding the color of the corresponding robot using local 'suc' and 'pred' (refer to~\cite{FPS19}).}.
\begin{figure}
\centering\includegraphics[keepaspectratio, width=14cm]{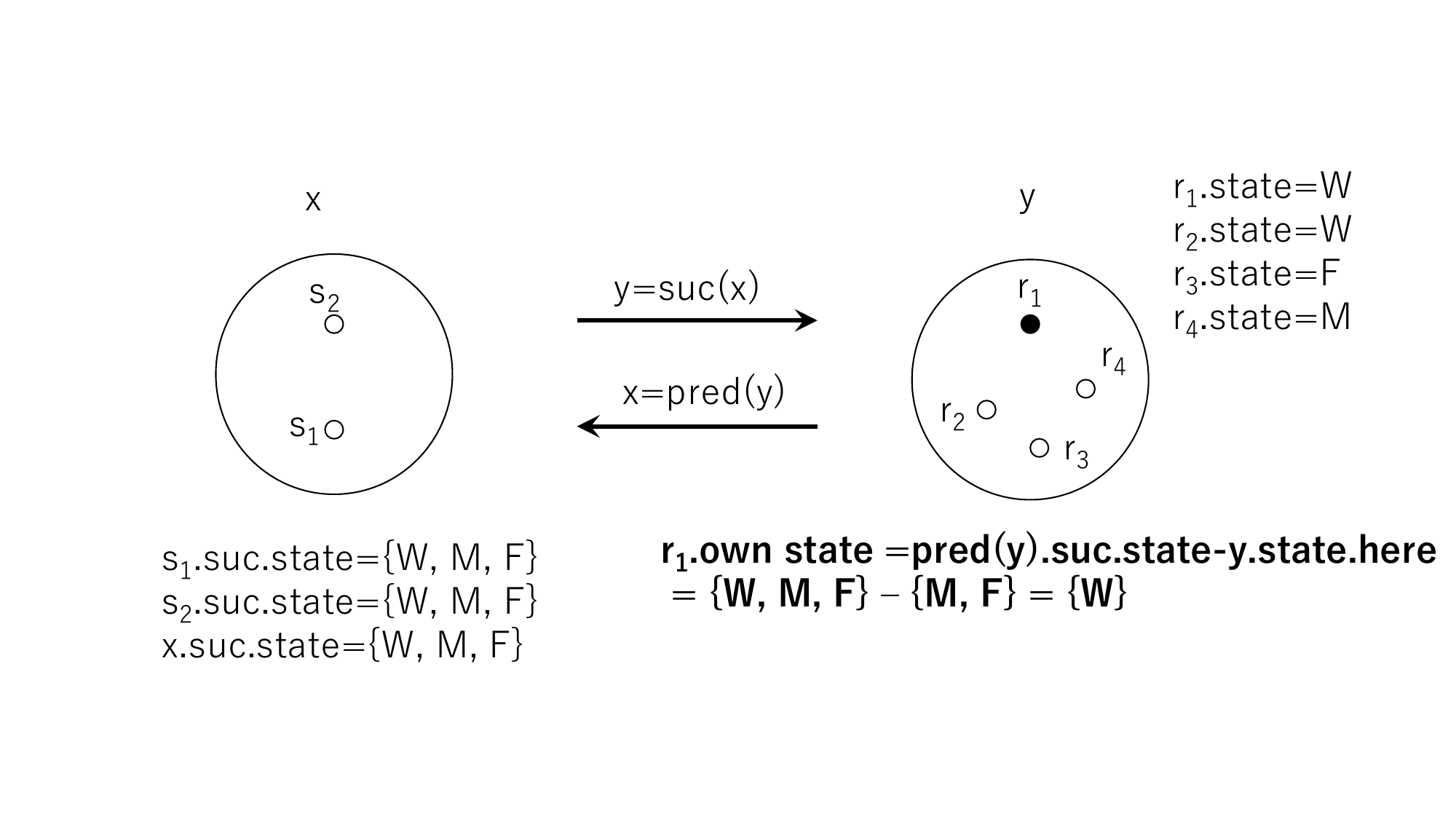}
    \caption{Determination of own.state of $r$.}
    \label{fig:determine-own-state}
\end{figure}
This phase consists of two stages: checking the end of mega-cycles and execution of simulation.

{\em Checking End of Mega-Cycles:}
The first part of this phase is to check the end of the current mega-cycle. 
Since $r.state=F$ means that robot $r$ has executed the simulation in the current mega-cycle, $Is$-$all(F)$ means that all robots have executed algorithm ${\cal A}$ in this mega-cycle\footnote{Since each robot recognize its own state in Phase~$1$, $Is$-$all(F)$ means $\forall r (r.state=F)$.}.  Then it moves to Phase~$m$ (resetting all state flags) and returns to  Phase~$1$.

{\em Perform Simulation:}  
If $Is$-$all(F)$ is false, an activated robot $r$ executes algorithm ${\cal A}$ and changes its $state$ to $M$, provided that $r$'s own state is $W$($r.state=W$)  and $r$ does not observe robots $\rho$ with $\rho.state=M$, which means some robot $\rho$ has executed algorithm ${\cal A}$ and set $\rho.flag$ to $M$. Since robots executing algorithm ${\cal A}$ change their $state$ to $M$, as long as activated robots do not observe robots $\rho$ with $\rho.state=M$, it is guaranteed that there is no possibility of observing moving robots and thus the simulated algorithm behaves under $CM$-{\bf atomic}-\ASY.
If robot $r$  observes robots $\rho$ with $\rho.state=M$, $r$ changes $r.phase$ to $2$.
Note that some robots $\rho$ with $\rho.state=M$ are moving.
The phase moves to $2$ after finishing all the execution of algorithm ${\cal A}$ via the mixed phases of $1$ and $2$, it can be guaranteed by changing all $phase$'s flags to $2$, that is reaching Phase~$2$.

\item {\bf Phase 2-Ensure the end of simulation and update the neighbor's state flag} ($\forall \rho(\neq r) [\rho.phase =2]$).

In the $\LK$-operation, $r$ understands to be in Phase~$2$ by observing   $\rho.phase=2$ of any other robots $\rho$. In this configuration,
there are no robots executing algorithm ${\cal A}$, no robot is moving, and the locations of the robots remain unchanged. 
An activated robot $r$ at $x$ only updates $r.\CoS.state$ by observing $\CoS(x)$ and change its {\em phase} flag form $2$ to $3$. When every robot
$r$ has $r.phase$ to $3$, the third phase starts.

\item {\bf Phase 3-Change flags from $M$ to $F$}   ($\forall \rho(\neq r) [\rho.phase =3]$). 

In the $\LK$-operation, $r$ understands to be in Phase~$3$ by observing $\rho.phase=3$ of any other robot $\rho$. 
In Phase~$3$, 
if robot $r$ has executed algorithm $\mathcal{A}$ in Phase~$1$ ($r.state=M$), then it
 changes its state flag from $M$ to $F$ (to insure  that the scheduling of robots performing the simulated algorithms is fair). After all robots with $M$ change their state flags to $F$,  every robot copies its neighboring states' flags ($\CoS(x).state$) 
setting Phase to $1$.

 

\item {\bf Phase $m$-Reset Mega-Cycle}   ($\forall \rho(\neq r) [\rho.phase =m]$). 
In the $\LK$-operation, $r$ understands to be in Phase~$m$ by observing $\rho.phase=m$ of any other robot $\rho$. 
In Phase~$m$, 
each robot $r$ sets $r.state=W$ and $r.\CoS.state  = \{W\}$
and,  after all robots reset their states flags,
the phase returns to $1$ to begin a new mega-cycle. Configurations having phase flags $m$ and $1$ occur in the cases from $Is$-$all$-$phases(1)$ to $Is$-$all$-$phases(m)$ and from $Is$-$all$-$phases(m)$ to $Is$-$all$-$phases(1)$. But 
it is not difficult for the robots to distinguish the particular transition being observed: if $Is$-$all(F)$ is true,  it is the former,  otherwise ($Is$-$all(W)$ is true) it is the latter.
(see Fig.~\ref{fig:transition-sim-SSY-by-ASY-LC-in-FCOM}). 

\end{itemize}




%


We prove the correctness of {\tt SIM(A)} working on \ASY.

Since we consider $\mathcal{FCOM}$ robots, when a robot $r$ checks a predicate, for example $\forall \rho(\rho.phase =\alpha)$,  $r$ cannot see its own $r.phase$. Then, $r$ only checks $\forall \rho\neq r(\rho.phase =\alpha)$ and observes that the predicate may be satisfied although $r.phase$ is not $\alpha$. 
Therefore, 
predicates appearing in the $\FC$ algorithm must be of the form 
$\forall$ $\rho\neq $r ($\ldots$)  and we must consider configurations on which only one robot $r$ observes that some predicate holds but any of the other robots observe that the predicate does not hold.
If it holds $\forall \rho(\rho.phase =\alpha)$ on the configuration, it is denoted by $same(phase=\alpha)$ (because all robots are in the same phase).  If it holds $\forall \rho \in R -\{r_e\} (\rho.phase =\alpha)$ and $r_e.phase=\beta$ on configuration, it is denoted by $except1(phase=\alpha;\beta(r_e))$ (because all robots except $r_e$ have  phase $\alpha$ and $r_e.phase=\beta$).  

The transition from Phase~$\alpha$ to $\gamma$ begins in configuration satisfying $same(step=\alpha)$ or $except1(step=\alpha;\beta(r_e))$
and ends in one satisfying $same(step=\gamma)$ or $except1(step=\gamma;\alpha(r'_e))$, where $(\alpha, \gamma)=(1,2),(2,3),(3,1),(1,m),$ and $(m,1)$.
In the configuration of $same(step=\alpha)$ or $except1(step=\alpha;\beta(r_e))$,
some robot changes its $phase$ to $\gamma$ at time $t$ and then the number of robots with $phase=\gamma$ increase. Finally, the configuration becomes one with $same(step=\gamma)$ or $except1(step=\gamma;\alpha(r'_e))$ at $t'$. Since the simulation algorithm works in \ASY, all robots are not inactive at the times $t$ and $t'$ in general. However, we can consider these times as if the start times (called pseudo start time, or ps-time) in the followings. Note that all robots do not move between $t+1$ and $t'$ in the algorithm. 

Let $t_{C}$ be the time $Comp$-operations are performed at which the number of robots with $phase=\alpha$ is at most one and let $R_C$ be a set of robots perform $Comp$-operations at $t_{C}$. Let $t'_{C}$ the time 
$Comp$-operations are performed just before $t_{C}$. There are two cases we consider.

(1) When the number of robots with $phase=\alpha$ is zero at $t_{C}$, robots in $R_C$ do not move and even if robots in $\R-R_C$ are activated in the time interval  $[t'_{C}+1..t_{C}]$, the lights of these robots are unchanged until they finish their {\textit LCM}-cycle. 

(2) When the number of robots with $phase=\alpha$ is one at $t_{C}$, let $r'_e$ be the robot with $phase=\alpha$. Note that $r'_e$ has not been activated between $t$ and $t_{C}$. Robots in $R_C$ also do not move and even if robots in $\R-(R_C \cup \{r'_e\})$ are activated in the time interval  $[t'_{C}+1..t_{C}]$, the lights of these robots are unchanged until they finish their LCM-cycle.

Then the time $t_{C}+1$ can be considered as all robots are inactive and robots can start in the configuration at that time.

Let $C_t$ be the configuration at $t$ and let $r$ be located at $x$ on $C_t$. We first consider a precondition to hold at the beginning of Phase~$1$ at $t$.

%

\begin{description}
\item[$PC_1(t)$:] $same(phase=1)$ \\ and $\forall r \in \R[ (r.state = W$ or $r.state=F)$  and  $r.\CoS.state=\CoS(x).state$] 
on $C_t$, or\\
$except1(phase=1;\alpha(r_e)(\alpha=3,m))$ and $\forall r \in \R [r.state = W$ or $r.state=F$] and $\forall r \in \R -\{r_e\}[r.\CoS.state=\CoS(x).state$] on $C_t$.\\
%
\end{description}

Note that the initial  configuration at time $0$ satisfies $PC_1(0)$ and $Is$-$all(F)=false$, in fact, for every robot $\rho$, $\rho.phase=1$, $\rho.state=W$ and $\rho.\CoS.state=\{W\}$.
The configuration satisfying the precondition $PC_1(t)$ occurs at the initial configuration, after one execution of the simulation and after any Mega-cycle is finished. The former two cases satisfy $Is$-$all(F)=false$ and the last case satisfies $Is$-$all(F)=true$.
The last case will return to the initial configuration with $Is$-$all(F)=false$ (Lemma~\ref{lem:mto1}).

\noindent
{\bf Case 1: Mega-cycle has been finished}

First, we consider the case that $PC_1(t)$ holds and $Is$-$all(F)=true$ at $t$, that is Mega-cycle is finished at $t$. Note that in this case $\alpha=3$ in $PC_1(t)$. Generally if $Is$-$all(F)=false$, then  $\alpha=m$ in $PC_1(t)$, otherwise, $\alpha=3$ in $PC_1(t)$ (this case occurs after one execution of the simulation).

(1-I) Consider that $except1(phase=1;\alpha=3(r_e))$ holds at $t$. If robot $r$ except $r_e$ is activated after $t$,
since $r$ observes $r_e.phase=3(\neq 1)$ and $Is$-$phase$-$mixed(1,3)$ hold, the configuration is not changed ({\bf lines}~28-30). Letting $t'$ be a time when $r_e$ is activated after $t$,
there is a time $t'_C$ such that  $r_e.\CoS.state$ is correctly set at $t'_C$ \footnote{$r_e.phase$ is also set to $1$. However, in this case, since $Is$-$all(F)=true$, $r_e.phase$ is set to $m$ after all ({\bf lines}~4-5).}.
In addition, $r_e$ observes $Is$-$all(F) =true$, $r_e$ sets $r_e.phase=m$ at $t'_C$.  Then since robot $r$ except $r_e$ executes {\bf line}~38\footnote{If $r_e$ observes $\rho.phase=m$ for some robot $\rho$, $r_e$ also executes {\bf line}~38.    Otherwise, $r_e$ observes $Is$-$phases(1)=true$ and $Is$-$all(F) =true$ and $r_e$ sets $r_e.phase=m$}, the number of $phase=m$ increases and number of $phase=1$ decreases after $t'_C$. 
Then
letting $t_m-1$ be the time $Comp$-operations are performed at which the number of robots with $phase=1$ is at most one,
$t_m$ becomes the ps-time such that $same(phase=m)$ or $except1(phase=m;1(r'_e)))$ holds at $t_m$.
Note that $Is$-$all(F)=true$ at time $t_m$.

(1-II) In the case that $same(phase=1)$ holds and $Is$-$all(F)=true$ at $t$,
since there is a robot that changes its $phase$ to $m$,
similarly we can show that there exists a ps-time $t_m$ such that $same(phase=m)$ or $except1(phase=m;1(r_e))$ holds at $t_m$.

Then the reset of \textit{state}s begins at $t_m$. 
There is a time such that $same(phase=m)$ holds and resetting \textit{state}s and $\CoS.state$s continues until there is at most one robot $r$ with $r.state=F$. If there is just one robot $r$ with $r.state=F$, any other robot than $r$ continue to reset and their phases remain $m$ since they observe $r.state=F$ ({\bf lines}~33-34).
On the other hand, since $r$ observes that there does not exist $F$ when activated after $t_m$, $r$ resets its own $state$ and $\CoS.state$ and changes $r.phase$ to $1$ and  ({\bf lines}~32,36).
In the case that there is no robot with $state=F$, the activated robots at that time change their $phase$ to $1$  and reset their own states and $\CoS.state$s\footnote{However, resetting has been already finished at this time.}.
In both cases, since there exists a robot $r$ with $r.phase=1$ and \textit{state}s and $\CoS.state$'s of all robots are reset (that is, $Is$-$all(W)$ is true),
the number of $phase=1$ increases and there exists a ps-time $t_1$ such that $PC_1(t_1)$ holds. Note that in this case $\alpha=m$ in $PC_1(t_1)$.

\begin{lemma}\label{lem:mto1}
Assume that $PC_1(t_0)$ holds  and $Is$-$all(F)=true$ at $C_{t_0}$, and the simulation algorithm is executed  from the configuration $C_{t_0}$ with $PC_1(t_0)$. 
Then there exists a ps-time $t_1$ such that the following conditions are satisfied for a configuration $C_{t_1}$;
\begin{description}
\item [(1)] $same(phase=1)$ and ($\forall r \in \R[ r.state =W$  and $\CoS(x).state=\{ W \}$])) on $C_{t_1}$, or\\

\item [(2)] $except1(phase=1;m(r_e))$ and ($\forall r \in \R[ r.state =W$  and $\CoS(x).state=\{ W \}]$)) on $C_{t_1}$\footnote{Note that this conditions satisfies $PC_1(t_1)$. In this case, $r_e$ resets $r_e.state$ and $r.\CoS.state$ at $t_1$.},

\noindent
where $x$ is the location occupied by robot $r$.
\end{description}
\end{lemma}

\noindent
{\bf Case 2: The simulation begins}

The following case is performing one execution of the simulation of algorithm \textit{A}. 
Let $t_1$ be a ps-time that satisfies $PC_1(t_1)$ and $Is$-$all(F)=false$.

First, we consider the case of $same(phase=1)$.

\begin{figure}

\centering\includegraphics[keepaspectratio, width=0.85\textwidth]{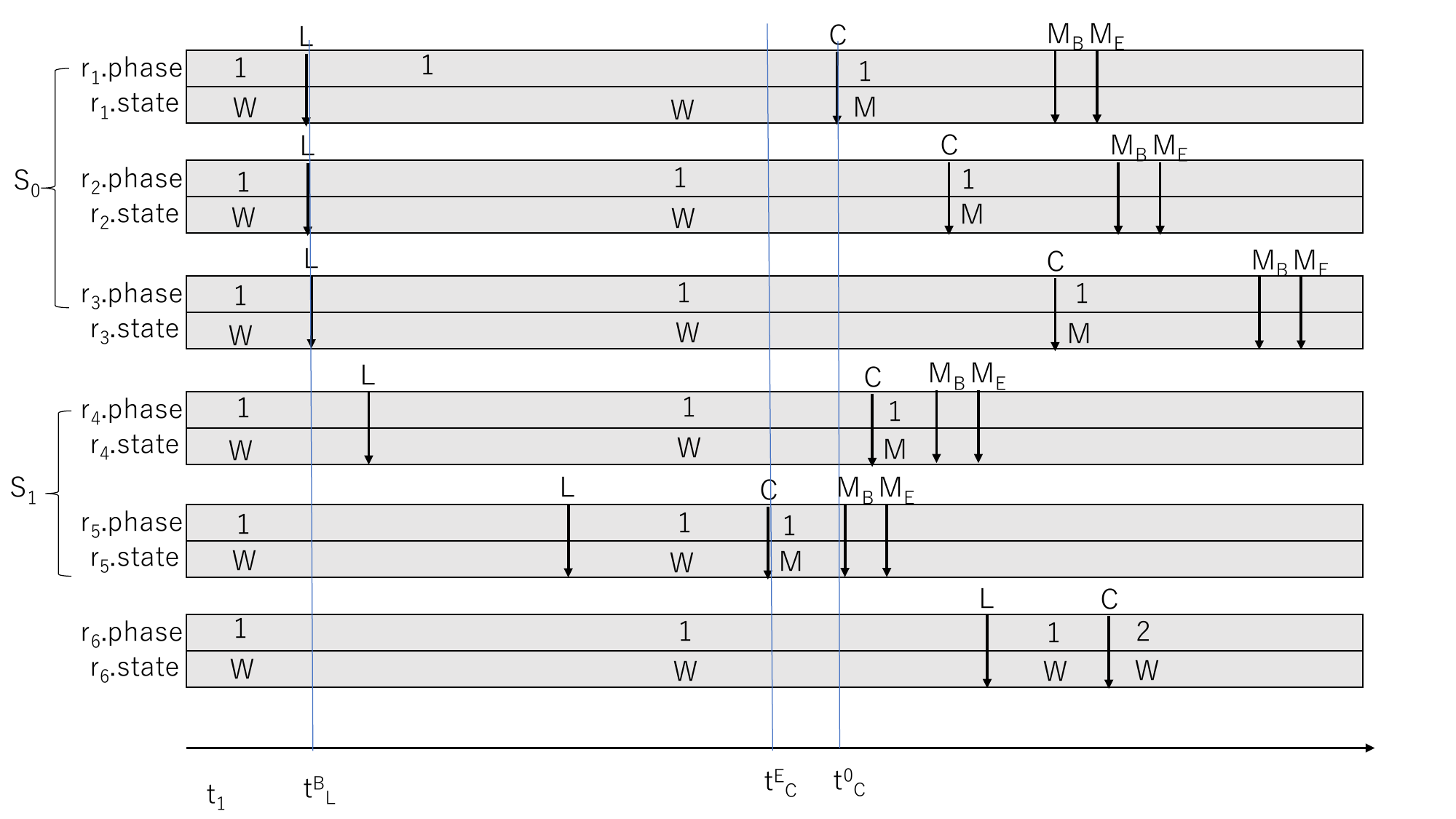}
    \caption{After $same(phase=1)$ and $Is$-$all(F)=false$ at $t_1$}
    \label{fig:S0S1}
\end{figure}

\noindent
(I) $same(phase=1)$ and $Is$-$all(F)=false$ at $t_1$:
Let $S_0$ be a set of the first activated robots  $\rho$ with $\rho.state=W$ after $t_1$ and
let $t^B_L$ be the  time when $Look$-operations of these robots are performed. 
Let $t^E_C$ be defined  as $\min \{ t | t$ is the time $\CP$-operation of robot in $S_0$ or the robot ($\not \in S_0$) with $\rho.state=W$ and activated after  $t^B_L$ is performed$\}$.
Let $S_1$ be a set of robots $\rho$ with $\rho.state=W$ and activated between $[t^B_L+1..t^E_C]$
and let $S=S_0 \cup S_1$.
Note that any robot in $S_0 \cup S_1$ does not observe robot $r$ with $r.state=M$ between $[t^B_L..t^E_C]$\footnote{Any robot $\rho$ with $\rho.state=F$ does nothing even if it is activated between $[t^B_L..t^E_C]$}, robots in $S_0 \cup S_1$ change their \textit{state}s to $M$ and execute algorithm~{\textit A} (see Fig.~\ref{fig:S0S1}).

We consider the two cases: (a) a robot, say $r'$ first activated after $t^E_C+1$ is not in $S$, (b) $r \in S$ is one of the first activated robots after $t^E_C+1$, and let $S'$ be a set of the first activated robots after $t_C^E+1$. Note that any robot in $S'$ has finished the first activation and its \textit{state} is $M$.

\noindent
(I-a) 
Note 
that for any robot $\rho$ in $S$ $\rho.phase=1$ until it finishes the execution of algorithm {\textit A} after $t^E_C+1$.  Letting $r''$ be a robot performing the $\CP$-operation at $t^E_C$, since $r''.state=M$ after $t^E_C+1$, $r'$ observes $r''.state=M$ and changes $r'.phase$ to $2$ ({\bf line}~6). 
After that, robots in $\R-S-\{r'\}$ change their phase flags to $2$ because they observe $\rho.phase=2$ for some robot $\rho$ and $Is$-$phase$-$mixed(1,2)$ is satisfied\footnote{$r'$ remains $r'.phase=2$ if it is activated.}. The robots in $S$ finish their simulation of {\textit A} and then change their phase flags to $2$. Then there exists  a ps-time $t'$ such that for the configuration $C_{t'}$ it holds that  $same(phase=2)$ or $except1(phase=2;1(r_e))$ and
if $\rho \in S$ then $\rho.state=M$ else $\rho.state$ is the same as that in $C_t$, and any robot in $S$ has completed its execution of the algorithm {\textit A} until $t'$. 

\noindent
(I-b)  If at least one robot in $S$ except $r$ performs the $Comp$-operation at $t_C^E$, the $state$ flag is $M$ at $t_E^C+1$. The robot $r$ observes  $state$ flags $M$ and changes $r.phase$ to $2$.
Thus, this case can be reduced to the case (a).
Otherwise, 
$r$ is the only robot in $S$ performing $Comp$-operation at $t_C^E$.
Although $r.state=M$ (but $r.phase=1$) at that time, $r$ observes the same snapshot as that at $C_t$ except for its own location. Then  $r$ begins executing the algorithm~{\textit A} again, because $r.own.state$ has not changed. If robots in $S$ except $r$ perform the $Comp$-operation at 
time $t''$ after that 
$r$ observes their  $state$s $M$ at $t''+1$ ($\exists \rho(\neq r )[(\rho.state=M)$ is true)  $r$ changes its phase flag to $2$. Thus, we can prove the case after $t''+1$ by using the method similar to the case (a),  there exists  a ps-time $t'$ such that for the configuration $C_{t'}$ it holds that  $same(phase=2)$ or $except1(phase=2;1(r_e))$ and
$\rho.state=M (\rho \in S)$ and for other robot $\rho' \in \R - S$ $\rho'.state$ is the same as that in $C_t$, robot in $S$ has completed its execution of the algorithm {\textit A} until $t'$. The difference is that if $r$ is activated $k$ times between $[t^E_C+1..t'']$, $r$ executes the algorithm~{\textit A} $k+1$ times. 

Noting that in the simulation algorithm
only the robots observing that there exist no $M$- $state$ flags in the configuration execute algorithm~\textit{A}. This means that the robots executing the algorithm do not observe other moving robots, that is the simulated algorithm obeys $\CM$-atomic \ASY.

Next, we consider the case $except1(phase=1;\alpha(r_e)(\alpha=3,m))$.

\noindent
(II) $except1(phase=1;\alpha(r_e)(\alpha=3,m))$ and $Is$-$all(F)=false$ at $t_1$: 
If $except1(phase=1;\alpha(r_e)(\alpha=3,m))$ holds, all robots except $r_e$ observe $Is$-$phase$-$mixed(1,\alpha)=true$ and their phase flags remain $1$
until $r_e$ is activated and performs $Comp$-operation. When $r_e$ is activated at $t'$ after $t_1$, $r_e$ observes $Is$-$all$-$phases(1)=true$ and $r_e$ changes $r_e.phase$ to $1$ and updates $r_e.\CoS.state$ correctly at time $t'_1+1$ ({\bf lines}~3-4), where $t'_1$ is the time when $\CP$-operation of $r_e$ is performed. Thus,
it can be reduced to the case (I).


Therefore, the following lemma holds.

\begin{lemma}\label{lem:1to2}
Assume that $PC_1(t_1)$ is satisfied and the simulation algorithm is executed from the configuration $C_{t_1}$ with $PC_1(t_1)$. 
Then there exists a ps-time $t_2$ such that the following conditions are satisfied for a configuration $C_{t_2}$;
\begin{description}
\item [(1)] $same(phase=2)$ {\bf or} $except1(phase=2;1(r_e))$,
\item [(2)] Let $S$ be a set of robots executing algorithm~{\textit A} between $[t_1..t_2]$. Then $S \neq \emptyset$ and any robot in $S$ does not observe moving robots (that is the $\CM$-{\bf atomic} condition is satisfied). And if $r \in S$ then $r.state=M$ else $r.state$ is the same as in $C_{t_1}$.
\end{description}
\end{lemma}

Let $PC_2(t_2)$ define the conditions satisfying  Lemma~\ref{lem:1to2} at ps-time $t_2$
and let $X(t_2)=(x_0(t_2), \ldots, x_{m_2}(t_2))$ be locations robots occupy in $C_{t_2}$. 
Note that $\CoS.state$s are not updated for $X(t_2)$  at time $t_2$.
If $except1(phase=2;1(r_e))$ holds, all robots except $r_e$ observe $Is$-$phase$-$mixed(1,2)=true$ and their phase flags remain $2$.
until $r_e$ is activated. When $r_e$ is activated at $t'$ after $t_2$, $r_e$ observes $Is$-$all$-$phases(2)=true$ and $r_e$ changes $r_e.phase$ to $3$ and updates $r_e.\CoS.state$ correctly. Therefore, $same(phase=2)$  holds after $t'+1$ and each robot $r$ activated after $t'+1$  changes $r.phase$ to $3$ and updates $r.\CoS.state$ correctly. Thus, the following lemma holds.
  
\begin{lemma}\label{lem:2to3}
Assume that $PC_2(t_2)$ is satisfied and the simulation algorithm is executed from the configuration $C_{t_2}$ with $PC_2(t_2)$. 
Then there exists a ps-time $t_3$ such that the following conditions are satisfied for a configuration $C_{t_3}$;
\begin{description}
\item [(1)] $same(phase=3)$ holds  or $except1(phase=3;2(r_e))$ holds at $t_3$, 

\item [(2)] For any robot $r \in \R$ at $x$, $r.state$ at $t_3$ is the same as at $t_2$.
\item [(3)] For any robot $r \in \R-\{r_e\}$ at $x$, $r.\CoS.state$ at $t_3$ is correctly set, that is, $r.\CoS.state=\CoS(x).state$.
\end{description}
\end{lemma}
Note that in Lemma~\ref{lem:2to3} if $same(phase=3)$ holds, $r.\CoS.state$ at $t_3$ is correctly set for any robot $r \in \R$ at $x$.

Let $PC_3(t_3)$ define the conditions that satisfy  Lemma~\ref{lem:2to3} at the time $t_3$
and let $X(t_3)=(x_0(t_3), \ldots, x_{m_3}(t_3))$ be locations robots occupy in $C_{t_3}$.
Note that $X(t_3)=X(t_2)$ because any robot does not move between $t_2$ and $t_3$.

If $except1(phase=3;2(r_e))$ holds, all robots except $r_e$ observe $Is$-$phase$-$mixed(2,3)=true$ and their phase flags and $\CoS.state$ flags remain $3$ and unchanged, respectively,
until $r_e$ is activated and its $\CP$-operation is performed. When $r_e$ is activated at $t'$ after $t_3$, $r_e$ observes $Is$-$all$-$phases(3)=true$ and $r_e$ changes $r_e.phase$ to $3$ and updates $r_e.\CoS.state$ correctly.  Then for any robot $r$ $r.\CoS.state$ is correctly set at time $t''+1$, where $t''$ is the time when the $\CP$-operation of $r_e$ is performed.
After $t''+1$, robot $r$ with $r.state=M$ changes it to $F$ while its $\CoS.state$ is updated\footnote{If $r_e.sate=M$, $r_e$ changes it to $F$ at time $t''+1$}. Since the number of $M$ appearing in the configuration is monotonically decreasing, there is a time when there is no $M$ in the configuration and some robot observing the configuration changes its $state$ flag to $1$. 
Then there exists a ps-time $t'''>t''$ such that 
(1) $same(phase=1)$ {\bf or} $except1(phase=1;3(r_e))$ at $t'''$, (2) $C_{t'''}$ has no $M$ in $state$ and $\CoS.state$ flags, and (3) All $\CoS.state$ flags are correctly set. 

In the case that  $same(phase=3)$ holds, since activated robot $r$ with $r.phase=3$ change it to $F$, above (1)-(3) also hold similarly. 
\begin{lemma}\label{lem:3to1}
Assume that $PC_3(t_3)$ holds and the simulation algorithm is executed from the configuration $C_{t_3}$ with $PC_2(t_3)$. 
Then there exists a ps-time $t_1$ such that the following conditions are satisfied for a configuration $C_{t_1}$;
\begin{description}
\item [(1)] $same(phase=1)$ holds or $except1(phase=1;3(r_e))$ holds at $t_3$, 

\item [(2)] For any robot $r \in \R$ at $x$, if $r.state=M$ at $t_3$ then $r.state=F$ at $t_1$, otherwise $r.state$ at $t_1$ is the same as that at $t_3$.
\item [(3)] For any robot $r \in \R-\{r_e\}$ at $x$, $r.\CoS.state$ at $t_1$ is correctly set, that is, $r.\CoS.state=\CoS(x).state$.
\end{description}
\end{lemma}

It is easily verified that $PC_1(t_1)$ holds by Lemma~\ref{lem:3to1}. Then the next simulation can be performed from $t_1$.  Therefore, by Lemmas~\ref{lem:1to2}-\ref{lem:3to1}, {\tt SIM(A)} executes Phase~$1$-Phase~$3$ and Phase~$m$ in infinite cycles in \ASY\ and the execution of \textit{A} obeys $CM$-{\bf atomic}-\ASY.
Let $E$ be the sequence of the set of activated robots that execute simulated algorithm~{\textit A}  in  Algorithm~{\tt SIM(A)}. 
Since,  by Lemma~\ref{lem:1to2},  any mega-cycle is completed, we can show that $E$ is fair.
Then we have obtained Theorem~\ref{th:FCACM=FCA}. Note that, if   algorithm~\textit{A} uses $\ell$ colors, the simulating algorithm~{\tt SIM(A)} uses $O(\ell)$ colors.

\subsubsection{Proof of Theorem \ref{th:FCALC=FCS}} 
\label{sec:FCALC=FCS}

We can use the same simulation algorithm to show the equivalence between $A_{LC}$ and \SSY\ in $\FC$.
%
%

\begin{lemma}\label{lem:1to2inALC}
Assume that $PC_1(t_1)$ is satisfied and the simulation algorithm is executed from   configuration $C_{t_1}$ with $PC_1(t_1)$ in $LC$-{\bf atomic}-\ASY. 
Then there exists a ps-time $t_2$ such that the following conditions are satisfied for a configuration $C_{t_2}$;
\begin{description}
\item [(1)] $(same(phase=2)$ or $except1(phase=2;1(r_e))$,
\item [(2)] Let $S$ be a set of robots executing algorithm ${\cal A}$ between $[t_1..t_2]$. Then $S \neq \emptyset$ and any robot in $S$ observes the same snapshot (that is, the \SSY\ condition is satisfied). Moreover,  if $r \in S$ then $r.state=M$; otherwise $r.state$ is the same as in $C_{t_1}$.
\end{description}
\end{lemma}
\begin{proof}
If the simulation algorithm works in $LC$-{\bf atomic}-\ASY,
$t^B_L=t^E_C$ in the proof of Lemma~\ref{lem:1to2} and $S=S_0$.  The robot $r \in S_0$ then sets $r.state=M$ at time $t^B_L+1$ and observes the same snapshot. Similarly to the proof of Lemma~\ref{lem:1to2}, there exists a ps-time $t_2$ such that for the configuration $C_{t_2}$ we have $(same(phase=2)$ or $except1(phase=2;1(r_e))$ and
if $\rho \in S$ then $\rho.state=M$ otherwise $\rho.state$ is the same as in $C_{t_1}$, and any robot in $S$ observes the same snapshot.
\end{proof}

\section{The \texorpdfstring{$\FS$}{FS}  Computational Landscape}
\label{sec-FS}

\subsection{Separating \SSY\ from \ASY\  in \texorpdfstring{$\FS$}{FS}}

In this section, we consider the $\FS$ model; in this model,
the only difference with  $\OB$ is that the robots are endowed 
with a bounded amount of memory whose content persists
from a cycle to the next.
We investigate whether, with this additional
capability, the  robots are able
to  overcome the limitations 
imposed by asynchrony,

The answer is unfortunately negative:
we prove that, also in this model,
 the otherwise enhanced robots are
 strictly more powerful 
  under the synchronous scheduler \SSY\ than
 under the
 asynchronous one \ASY.

To do so, we consider the problem \MLCv\ again.

Observe that \MLCv\ can be solved even in $\OB^S$ (Lemma~\ref{lem:CNV-WOC1}), and thus
in $\FS^S$.

\begin{lemma}\label{lem:MLC-pos}
$\MLCv \in \FS^{S}$;
this holds even under 
variable disorientation, 
non-rigid movement and in absence of chirality.
\end{lemma}
On the other hand, \MLCv\ cannot be solved  in $\FS^{A_M}$.

\begin{lemma}\label{lem:MLC-nega}
$\MLCv\notin \FS^{A_M}$
\end{lemma}
\begin{proof}
Let $r$ and $q$ be the two robots that we consider. 
In what follows, we will show that for any algorithm, the adversary can activate $r$ and $q$ 
and exploit variable disorientation so that they violate the condition of $\MLCv$.

Because of variable disorientation,  whenever a robot $X \in \{r,q\}$ performs a Look operation, the adversary can (and will)  force the {\em observed}
distance between $r$ and $q$  in the resulting snapshot to be always 1 (i.e., equals the 
current unit distance of $X$). Let $f(c, d)$ be the length of the computed move when a robot has color $c$ and the {\em real} distance between the two robots is $d$ in the last Look phase. Note that $F(c) = f(c,d)/d$ does not depend on $d$ because the distance always looks one to the robots.

Since the distance always looks the same to the robots, unless the two robots meet, the transition 
sequence of the  internal colors set by a robot 
is fixed. In particular, since the number of colors is a fixed constant, 
after a finite transient, say $(c_0, c_1, \dots, c_k)$, the sequence becomes periodic, say $(c_s, c_{s+1}, ..., c_k)^*$.

Then, the adversary can activate the robots in the following way so that, either during the transient  they violate the condition of $\MLCv$, or both of them end the transient without meeting each other and have color $c_s$:

\begin{enumerate}
\item $i \gets 0$.
\item 
If $F(c_i) > 1/2$, the adversary activates both $r$ and $q$, by which they pass each other, clearly violating the condition of $\MLCv$.
If $F(c_i) \le 1/2$, the adversary first activates $r$, and then activates $q$, by which $r$ and $q$ never meets
(i.e., never reach the same location).
\item $i \gets i + 1$ and go back to 2.
\end{enumerate}

If the robots did not violate  the condition of $\MLCv$ during 
their transient, they are both at the beginning of their periodic sequence with color $c_s$ in distinct positions. 
If $F(c_i)=0$ holds for all $i=s,s+1,\dots,k$, no robot moves, thus $\MLCv$ is never solved.
So, without loss of generality, we assume $c_s = c_0$ and $F(c_0) > 0$.
Then, the following strategy of the adversary leads to the violation of the condition of $\MLCv$,
where $d_0$ is the distance between $r$ and $q$ at time $0$.
\begin{enumerate}
 \item Let $r$ and $q$ perform $Look$ and $Compute$ phase, by which both $r$ and $q$ compute to move by distance $f(c_0, d_0)$.
 \item While $r$ is still waiting to be activated to move, activate only $q$ repeatedly until $q$ overtakes $r$ or the distance between $r$ and $q$ becomes less than $f(c_0, d_0)$. The former case occurs if $F(c_i) > 1$ for some $i$. This obviously violates the condition of $\MLCv$. Otherwise, the latter case must eventually occur because the distance between $r$ and $q$ becomes constant times smaller each time $q$ changes its color $k$ times.
 Then, the adversary finally activate $r$ to perform its Move phase. Then, $r$ moves a distance $f(c_0, d_0)$ and overtakes $q$, violating the condition. 
\end{enumerate}
Thus, for any algorithm, the two robots must violate the condition of $\MLCv$. 
\end{proof}

Thus, by Lemmas  \ref{lem:MLC-pos} 
 and \ref{lem:MLC-nega},
a separation between \SSY\ and \ASY\ in $\FS$ is shown.

\begin{theorem}\label{th:domFSOBAMoverFSOBA}
$\FS^{S} >\FS^{A}$  
\end{theorem}

\subsection{Refining the \texorpdfstring{$\FS$}{FS} landscape}\label{sec:refine-FSTA}
We can refine the $\FS$ landscape as follows;
Consider again the TF problem defined and analyzed in Section \ref{sec: OBlandscape}.
By Lemma \ref{lem:TF1}, TF can be solved  in $\OB^{A_M}$, and thus in
$\FS^{A_M}$.

On the other hand, TF
is not solvable in $\FS^{A_{LC}}$.

\begin{lemma}\label{lem:TF2}
$\text{TF} \not \in \FS^{A_{LC}}$, even with fixed disorientation.   
\end{lemma}
\begin{proof}
By contradiction, let ${\cal A}$ be an algorithm that always allows the two  $\FS$ robots
to solve TF and form a trapezoid reaching a terminal state in finite time under the $A_{LC}$ scheduler.
Consider  the initial configuration
where  $a$ is further than $b$ from $\overline{CD}$, and   
$\alpha = \pi/4$. 
Starting from this configuration, $a$ is  required to  move within finite time along $Y(A)$; 
on the other hand,  no other robot is  allowed to move.
Consider now the execution of ${\cal A}$ in which only
 $a$ is activated, and starts moving at time $t$; observe that,
 as soon as $a$  moves,  it creates a configuration where $a$ is still  further than $b$ from $\overline{CD}$,  but  
  $\alpha' = \min \{ \angle{b(t)a(t)A'}, \angle{a(t)b(t)B'} \} < \pi/4$.
  
  Activate now $b$ at time $t'>t$
while $a$ is still moving. Should this have been an initial configuration, within  a constant number of
activations (bounded by the number of internal states),
$b$ would move, say at time $t"$. In the current execution,
slow down the movement of $a$ so that it is still moving at time $t''$.
Since  in $\FS$\ $b$ cannot access the internal state of $a$, 
 nor remember previously observed angles and distances,
 it cannot detect that the observed configurations are not   initial configurations;
 hence it will move at time $t"$, violating  $TF2$ and contradicting the assumed correctness 
 of  ${\cal A}$. 
\end{proof}

Summarizing: by definition, $\FS^{A_{M}}\  \geq \FS^{A}$; by Lemma \ref{lem:TF1}, it follows
that TF is solvable in $\FS^{A_{M}}$; and, by
Lemma \ref{lem:TF2}, it follows
that TF is not solvable in $\FS^{A}$. In other words:
\begin{theorem}
\label{thm:TF-FS}
$\FS^{A_{M}} > \FS^{A}$
\end{theorem}


\begin{theorem}
\begin{enumerate}
\item $\FS^{A_{LC}}\  \equiv \FS^{A}$
\item $\FS^{S}$ > $\FS^{A_{M}}$
\item $\FS^{S}\   > \FS^{A_{LC}}$
\item $\FS^{A_{M}}\ > \FS^{A_{LC}}$ 
\end{enumerate}
 \end{theorem}
\begin{proof} {\bf 1.} holds because, by definition, $\FS$ robots cannot distinguish between $A_{LC}$ and $A$.
{\bf 2.}
follows  from follows from Lemmas~\ref{lem:MLC-pos} and~\ref{lem:MLC-nega}. 
{\bf 3.} follows from {\bf 1.} and Theorem \ref{th:domFSOBAMoverFSOBA}. {\bf 4.} follows from {\bf 1.} and
Theorem \ref{thm:TF-FS}.
\end{proof}

 \section{Relationship Between Models Under Asynchronous Schedulers}
In the previous sections, we have characterized the asynchronous landscape within each robot model. In
this section, we  determine the computational relationship between the different models under the  asynchronous schedulers $A_{LC}, A_M,A_{CM}$ and \ASY.

We do so by first determining the relationship
between  $\FC$ and the other models under the asynchronous schedulers; we then 
complete the characterization of the landscape by
establishing the still remaining relationships,
those between  $\FS$ and $\OB$.


\subsection{Relative power of \texorpdfstring{$\FC$}{FC}} 
In this section, we determine the relationship
between  $\FC$ and the other models under the asynchronous schedulers $A_{LC}, A_M,A_{CM}$ and \ASY.

We first  show  that
 $\FC^{A_{LC}}$ and $\FS^{A_{M}}$ are orthogonal. 
 To prove this result we use the existence of a problem,   {\tt Cyclic Circles}  (CYC),
 shown  in \cite{BFKPSW22} to be solvable in $\FC^{A}$ but not in $\FS^{S}$:
 
 \begin{lemma}\em{\cite{BFKPSW22}}\label{lem:CYC}
\begin{enumerate}
    \item  $\text{CYC} \not \in \FS^{S}$ 
    \item  $\text{CYC}  \in \FC^{A}$,  
     even under non-rigid-movement.
\end{enumerate}
\end{lemma}

We then consider the problem  {\tt Get Closer but Not too Close on Line} (GCNCL) defined as follows. 

\begin{definition}
{{\bf Get Closer but Not too Close on Line} (GCNCL):
Let $a$ and $b$ be two robots on distinct locations $a(0), b(0)$
where $r(t)$ denotes the position of  $r\in\{a,b\}$ at time $t\geq 0$.
This problem
requires the two robots to get closer, without ever increasing their distance on the
line connecting them, and eventually stop at distance at least $|a(0)-b(0)|/2$ from each other.}
\end{definition}

In other words, an algorithm solves GCNCL iff it satisfies the following predicate:
\begin{align*}
GCNCL \equiv
&\left(\forall t\geq 0: a(t),b(t)\in\overline{a(0)b(0)}\right )
\land \left(\forall t,t': 0 \le t \le t' \to d_{t} \ge d_{t'}\right)\\
&\land \left( \exists t: \frac{d_0}{2} \le d_t < d_0 \land (\forall t'\ge t: a(t)=a(t') \land b(t)=b(t') )\right),
\end{align*}
where $d_t$ is the distance between the two robots at time $t$, i.e., $d_t=|a(t)-b(t)|$.

\begin{lemma}
\label{lem:SMLS}\
\begin{enumerate}
    \item  $\text{GCNCL} \not \in \FC^{S}$.
    \item  $\text{GCNCL}  \in \FS^{A}$.
\end{enumerate}
\end{lemma}
\begin{proof}
{\bf 1.}  The impossibility of  $\FC^{S}$ can be obtained as follows.
Since we consider $\FC$, a robot computes its destination depending on the color of its opponent, not on its own color.
We say that a color $c$ is \emph{attractive} if a robot decides to move (i.e., not stay) 
when the color of the opponent is $c$.
The adversary can prevent the robots from solving GCNCL in the following way.
Initially, both robots have the same color. 
If that color is not attractive, the adversary keeps on simultaneously activating both robots until the color of the robots becomes attractive. During this period, no robot moves by the definition of attractive colors.
Note that an attractive color must appear eventually to solve GCNCL. 
From then on, the adversary keeps on activating only one robot, say $a$, while never activating $b$.
During this period, $b$ never changes its color, so the color of $b$ is always attractive. 
Because of variable disorientation, the adversary can guarantee that 
there is a fixed positive constant $c\le 1$ such that when $a$ is activated at time $t$, 
the resulting distance between $a$ and $b$ (after $a$ moves) is $c\cdot d_t = c \cdot |a(t)-b(t)|$.
(The robots immediately violate the specification of GCNCL if $c > 1$ or $c=0$.)
However, this implies that the distance between $a$ and $b$ converges to zero as $a$ moves repeatedly,
violating the specification of GCNCL. 

{\bf 2.} The problem is easily solvable with $\FS$ robot in \ASY. Let the robots have color $A$ initially. The first time a robot is activated, it moves closer by distance $d/4$ to the other and changes its color to $B$, where $d$ is the observed distance.
Whenever a robot is activated, if its color is $B$, it does not move.
Clearly, both robots eventually stop and their final distance is at least $d_0/2$.
\end{proof}

The orthogonality of $\FC^{A_{LC}}$ and $\FS^{A_{M}}$ (or $\FS^{A}$) then 
  follows from Lemmas~\ref{lem:CYC} and \ref{lem:SMLS}.

\begin{theorem}
\label{th:powerFCALC}\
\begin{enumerate}
     \item $\FC^{A_{LC}} \bot $ $\FS^{A_{M}}$
     \item $\FC^{A_{LC}} \bot $ $\FS^{A}$
    \item $\FC^{A_{LC}} >$ $\OB^{S}$
\end{enumerate}
\end{theorem}
 \begin{proof}
 {\bf 1.}- {\bf 2.} By Lemmas~\ref{lem:CYC} and \ref{lem:SMLS}.
 {\bf 3.} is proved by the fact that RDV can be solved by $\FC^S$ but not by $\OB^{S}$, and by the equivalence of $\FC^S$ and $\FC^{A_{LC}}$.
\end{proof}
 
 The following theorem shows  the relative power of $\FC^{A}$.
 
 \begin{theorem}
 \label{th:powerFCA}\
\begin{enumerate}
     \item $\FC^{A}(\equiv \FC^{A_{M}}) \bot $ $\FS^{A_{M}}$
     \item $\FC^{A} \bot $ $\FS^{A}$
     
     \item $\FC^{A} \bot $ $\OB^{S}$
    \item $\FC^{A} >$ $\OB^{A_{M}}$
\end{enumerate}
\end{theorem}
 \begin{proof}
  {\bf 1.} (resp. {\bf 2.}) follows from Theorem~\ref{th:powerFCALC} {\bf 1.} (resp. {\bf 2.}) and noting that CYC can be solved in $\FC^{A}$. {\bf 3.} is proved by Lemmas~\ref{lem:CNV-WOC1} and~\ref{lem:MLCv-imp}(MLCv can be solved in $\OB^S$ but cannot be solved in $\FC^{A_M}$) and the fact that CYC cannot be solved in $\FS^S$ (and so $\OB^S$).  {\bf 4.} is proved by the equivalence of $\FC^{A_M}$ and $\FC^{A}$ and using the result of RDV.
 \end{proof}

\subsection{Completing the characterization:  \texorpdfstring{$\FS$}{FS} vs \texorpdfstring{$\OB$}{OB} }

The relationship between $\FC$ and the other models under the
asynchronous schedulers has been determined in the previous section (Theorems \ref{th:powerFCALC} and
\ref{th:powerFCA}). To complete the characterization of the relationship between the computational power of the models under the asynchronous schedulers, we need to determine the relationship between  $\FS$ and $\OB$.

 \begin{theorem}\label{th:powerFSAM}\
 \begin{enumerate}
     \item\label{FSAM-OBS} $\FS^{A_{M}} \bot$  $\OB^{S} $
      \item\label{FSAM-OBAM} $\FS^{A_{M}} > $ $\OB^{A_{M}} > \OB^{A}$
    \item\label{FSA-OBAM} $\FS^{A} \bot $ $\OB^{A_{M}}$
    \item\label{FSA-OBS} $\FS^{A} \bot $ $\OB^{S}$
    \item\label{FSA-OBA}  $\FS^{A} > $ $\OB^{A}$
\end{enumerate}
\end{theorem} 
 \begin{proof}
  Note that RDV can be solved in $\FS^A$ (and so $\FS^{A_M}$) but cannot be solved in $\OB^S$ (and so $\OB^{A_M}$ and $\OB^{A}$).
 {\bf \ref{FSAM-OBS}.} is proved by the results of RDV, 
 and MLCv, which can be solved in $\OB^{S}$ but cannot be solved in $\FS^{A_M}$ (Lemmas~\ref{lem:CNV-WOC1} and~\ref{lem:MLC-nega}). 
  {\bf \ref{FSAM-OBAM}.} is proved with the result of RDV and Theorem~\ref{thm:TFoblot}. {\bf \ref{FSA-OBAM}.} (resp.  {\bf \ref{FSA-OBS}.}) are proved with the result of RDV and TF (Lemmas~\ref{lem:TF1},~\ref{lem:TF2} and the equivalence of $\FS^{A_{LC}}$ and $\FS^A$) (resp. MLCv (Lemmas~\ref{lem:CNV-WOC1},~\ref{lem:MLC-nega} and Theorem~\ref{thm:TF-FS})). {\bf \ref{FSA-OBA}.} is proved by the result of RDV.
 \end{proof}
 

%
\section{Concluding Remarks}
\label{sec:conclusion}

In this paper, we investigated the computational relationship
between the power of the four models   $\OB$, $\FS$, $\FC$  and $\LU$,
under a range of asynchronous schedulers, from  \SSY\ to \ASY, and provided a complete characterization of such relationships.
In this process, we have  established a variety of results on the computational powers
of the robots in presence or absence of (limited) internal capabilities of memory persistence and/or communication. 
These results include  the proof of 
computational separation between  \SSY\ and \ASY\ in absence of either capability,
closing several important open questions.

This investigation has also provided valuable insights
into the elusive nature of the relationship between  asynchrony 
and the level of atomicity of  the $\LK$, $\CP$, and $\M$  operations performed in an $\mathit{LCM}$ cycle.
In fact, in this paper, the study of the asynchronous landscapes has focused on 
precisely the set of asynchronous schedulers defined by the different possible
atomic combinations of those operations as well as the $\M$  operation:
 starting from $LCM$-{\bf atomic}-\ASY,  which corresponds to \SSY, ending with \ASY, and including $LC$-{\bf atomic}-\ASY,  $CM$-{\bf atomic}-\ASY, and  $M$-{\bf atomic}-\ASY.

These results  open several new research directions. In particular, an important direction is the examination of other classes of asynchronous schedulers, to
further understand the nature of asynchrony for robots operating in $\mathit{LCM}$ cycles, identify the crucial factors that render asynchrony difficult for the robots,  and possibly discover new methods to overcome it.

%
%

\clearpage
\bibliographystyle{plain}
\bibliography{referencesorg}

\end{document}